\documentclass[11pt,a4paper,reqno]{amsart}%
\usepackage{amsthm,amsmath,amsfonts,amssymb,amsxtra,appendix,bookmark,dsfont,latexsym}

\makeatletter
\usepackage{hyperref}
\usepackage{delarray,a4,color}
\usepackage{euscript}
\usepackage[latin1]{inputenc}

\usepackage{pdfsync}
\makeatletter
\def\@tocline#1#2#3#4#5#6#7{\relax
  \ifnum #1>\c@tocdepth 
  \else
    \par \addpenalty\@secpenalty\addvspace{#2}%
    \begingroup \hyphenpenalty\@M
    \@ifempty{#4}{%
      \@tempdima\csname r@tocindent\number#1\endcsname\relax
    }{%
      \@tempdima#4\relax
    }%
    \parindent\z@ \leftskip#3\relax \advance\leftskip\@tempdima\relax
    \rightskip\@pnumwidth plus4em \parfillskip-\@pnumwidth
    #5\leavevmode\hskip-\@tempdima
      \ifcase #1
       \or\or \hskip 1em \or \hskip 2em \else \hskip 3em \fi%
      #6\nobreak\relax
      \dotfill
      \hbox to\@pnumwidth{\@tocpagenum{#7}}
    \par
    \nobreak
    \endgroup
  \fi}
\makeatother

\numberwithin{equation}{section}

\newtheorem{theorem}{Theorem}[section]
\newtheorem{lemma}[theorem]{Lemma}
\newtheorem{asumption}[theorem]{Assumption}
\newtheorem{corollary}[theorem]{Corollary}
\newtheorem{proposition}[theorem]{Proposition}

\theoremstyle{definition}
\newtheorem{definition}[theorem]{Definition}

\theoremstyle{remark}
\newtheorem{remark}[theorem]{Remark}

\usepackage{color}

\newcommand{\bdm}{\begin{displaymath}}
\newcommand{\edm}{\end{displaymath}}
\newcommand{\bdn}{\begin{eqnarray}}
\newcommand{\edn}{\end{eqnarray}}
\newcommand{\bay}{\begin{array}{c}}
\newcommand{\eay}{\end{array}}
\newcommand{\ben}{\begin{enumerate}}
\newcommand{\een}{\end{enumerate}}
\newcommand{\beq}{\begin{equation}}
\newcommand{\eeq}{\end{equation}}

\newcommand{\eps}{\varepsilon}

\newcommand{\norm}[1]{\left\lVert #1 \right\rVert}

\newcommand{\cZ}{\mathcal{Z}}
\newcommand{\R}{\mathbb{R}}
\newcommand{\N}{\mathbb{N}}

\newcommand{\C}{\mathbb{C}}

\newcommand{\cE}{\mathcal{E}}

\newcommand{\cH}{\mathcal{H}}
\newcommand{\bH}{\mathbb{H}}
\newcommand{\cW}{\mathcal{W}}

\newcommand{\cL}{\mathcal{L}}

\newcommand{\one}{{\ensuremath {\mathds 1} }}

\newcommand{\dd}{\partial}
\newcommand{\half}{\frac{1}{2}}

\newcommand{\intR}{\int_{\R ^2}}
\newcommand{\intRN}{\int_{\R ^{2N}}}

\newcommand{\LLL}{\mathfrak{H}}

\newcommand{\cLau}{c _{\rm Lau}}

\newcommand{\PsiLau}{\Psi_{\rm Lau}}

\newcommand{\rhoF}{\rho_F}

\newcommand{\muF}{\mu_F}

\newcommand{\TFM}{\mathcal{M} ^{\rm TF}}
\newcommand{\ETF}{\mathcal{E} ^{\rm TF}}

\newcommand{\TFmin}{\sigma ^{\rm TF}}
\newcommand{\TFpot}{\Phi ^{\rm TF}}
\newcommand{\STF}{\Sigma ^{\rm TF}}

\newcommand{\Vnuc}{V_{\rm nuc}}
\newcommand{\rhonuc}{\rho_{\rm nuc}}

\newcommand{\mubf}{\boldsymbol{\mu}}
\newcommand{\nubf}{\boldsymbol{\nu}}

\newcommand{\rhoM}{\varrho_{\mathrm{max}}}

\newcommand{\Phiav}{\bar{\Phi}}
\newcommand{\rb}{\bar{r}}

\newcommand{\cEt}{\tilde{\mathcal{E}}}

\title{Local incompressibility estimates for the Laughlin phase}

\author[E. H. Lieb]{Elliott H. Lieb}
\address{Departments of Mathematics and Physics, Princeton University,
Princeton, NJ 08544, USA}
\email{lieb@princeton.edu}

\author[N. Rougerie]{Nicolas Rougerie}
\address{Universit\'e Grenoble 1 \& CNRS ,  LPMMC (UMR 5493), B.P. 166, F-38 042 Grenoble, France}
\email{nicolas.rougerie@grenoble.cnrs.fr}

\author[J. Yngvason]{Jakob Yngvason}
\address{Fakult\"at f\"ur Physik, Universit{\"a}t Wien, Boltzmanngasse 5, 1090 Vienna, Austria \&
Erwin Schr{\"o}dinger Institute for Mathematical Physics, Boltzmanngasse 9, 1090 Vienna, Austria.}
\email{jakob.yngvason@univie.ac.at}

\date{January, 2018}

\begin{document} 

\begin{abstract}
We prove sharp density upper bounds on optimal length-scales for the ground states of classical 2D Coulomb systems and generalizations thereof. Our method is new, based on an auxiliary Thomas-Fermi-like variational model. Moreover, we deduce density upper bounds for the related low-temperature Gibbs states. Our motivation comes from fractional quantum Hall physics, more precisely, the perturbation of the Laughlin state by  external potentials or impurities. These give rise to a class of many-body wave-functions that have the form of a product of the Laughlin state and an analytic function of many variables. This class is related via Laughlin's plasma analogy to Gibbs states of the generalized classical Coulomb systems we consider. Our main result shows that the perturbation of the Laughlin state cannot increase the particle density anywhere, with implications for the response of FQHE systems to external perturbations.
\end{abstract}

\maketitle

\tableofcontents

\section{Introduction}\label{sec:intro}

The fractional quantum Hall effect (FQHE)~\cite{Girvin-04,Jain-07,Laughlin-99,StoTsuGos-99} is a remarkable feature of the transport properties of 2D electron gases under strong perpendicular magnetic fields and at low temperatures. Soon after its experimental discovery~\cite{TsuStoGos-82}, it was recognized in the seminal works of Laughlin~\cite{Laughlin-83,Laughlin-87} that the origin of the effect lies in the emergence of a new, strongly correlated, phase of matter. The latter has been argued to host elementary excitations with fractional charge, a fact that was later experimentally confirmed~\cite{SamGlaJinEti-97,MahaluEtal-97,YacobiEtal-04}. Even more fascinating, but still lacking an experimental confirmation, is the possibility that these excitations (quasi-particles) are anyons~\cite{AroSchWil-84,ZhaSreGemJai-14,LunRou-16}, i.e. have quantum statistics different from those of bosons and fermions, the only known types of 
fundamental particles. Due to these exciting prospects, it is an ongoing quest in condensed matter physics to generalize FQHE physics to other, more flexible, contexts than the 2D electron gas~\cite{ParRoySon-13,BerLiu-13,Cooper-08,BloDalZwe-08,Viefers-08}. 

One of the distinctive features of strongly correlated FQH states is that they are incompressible liquids, rendering them very robust against perturbations and external fields. Incompressibility, in the form relevant to the FQHE, has in fact two aspects:

\smallskip

\noindent \textbf{1.}~The proposed strongly correlated FQHE wave functions are approximate ground states for the many-body Hamiltonian of the system at hand. Their energy is separated from the rest of the spectrum by a gap independent of volume and particle number.

\smallskip 

\noindent \textbf{2.}~Modifications of model FQHE ground states that stay within the (highly degenerate) ground eigenspace of the many-body Hamiltonian cannot increase the local one-particle density beyond a fixed value.

\smallskip

These two aspects are very non-trivial to check, particularly since we are talking about strongly correlated states of matter for which mean-field descriptions in terms of independent particles are not adequate. The main evidence so far in favor of Properties~\textbf{1} and~\textbf{2} has been experimental and numerical and it is an important theoretical challenge to improve on this.
\medskip 

The main paradigms of the FQHE are encoded in Laughlin's wave-function, the simplest FQH state:
\begin{equation}\label{eq:laufunc}
\PsiLau (z_1,\ldots,z_N) =\cLau \prod_{i<j}(z_i-z_j)^{\ell} e^{-B\sum_{i=1}^N|z_i|^2/4}.
\end{equation}
Here $z_1,\ldots,z_N\in\mathbb C$ are the positions of $N$ particles moving in $\mathbb R^2$, identified with the complex plane, and the constant $\cLau$ is a normalization factor (setting the $L^2$-norm equal to $1$). For fermions, $\ell$ is odd and $\geq 3$ (the case $\ell=1$ corresponds to noninteracting fermions), while for bosons $\ell\geq 2$ is even. The function~\eqref{eq:laufunc} has originally been proposed as a variational ansatz for the ground state of the many-body magnetic Schr\"odinger Hamiltonian\footnote{We use the label QM to distinguish the quantum mechanical Hamiltonian~\eqref{eq:mag hamil} from the classical Hamiltonians $H_N$ that will be introduced later in Section~\ref{sec:proof out}}
\begin{equation}\label{eq:mag hamil}
H_N^{\rm QM} = \sum_{j=1} ^N \left[\left( -i \nabla_j - \frac{B}{2} x_j ^\perp \right) ^2 + V (x_j) \right]+ \sum_{1\leq i < j \leq N} w(x_i-x_j) 
\end{equation}
acting on $L ^2 (\R ^{2N}) $, the Hilbert space for $N$ 2D particles. Here $x^{\perp}$ denotes the vector $x\in \R^2$ rotated by $\pi/2$ counter-clockwise, so that 
$$ \mathrm{curl} \frac{B}{2} x ^\perp = B$$
and thus $\frac{B}{2} x ^\perp$ is the vector potential of a uniform magnetic field, expressed in symmetric gauge.

To arrive at an ansatz of the form~\eqref{eq:laufunc} one assumes that the energy scales are set, in order of importance, by
\begin{itemize}
 \item the strength $B$ of a constant applied magnetic field perpendicular to the plane,
 \item the \emph{repulsive} pair-interaction potential $w$,
 \item the external scalar potential $V$, representing trapping and/or disorder. 
\end{itemize}
In the sequel we choose units so that the strength of the magnetic field is 2 and the magnetic length therefore $1/\sqrt2$.

Since the magnetic field is the main player, the first reduction is to replace the state space $L^2 (\R^{2N}) \simeq \bigotimes ^N L^2 (\R^2)$ by $\bigotimes ^N \LLL$ where ($z$ is the complex coordinate in the plane)
\beq 
\LLL = \left\{ \psi (z) = f(z) e ^{-|z| ^2 /2} \in L ^2 (\R ^2) \:|\: f \mbox{ analytic }\right\}
\eeq
is the lowest Landau level (LLL), ground eigenspace of the magnetic Laplacian $\left( -i \nabla - \frac{B}{2} x ^\perp \right) ^2$. In other words, one reduces available one-body orbitals to those minimizing the magnetic kinetic energy. Spins are all aligned with the external magnetic field so that one only needs consider symmetric/antisymmetric wave-functions of $\bigotimes ^N \LLL$ to describe bosons/fermions.

Secondly, in order to suppress repulsive interactions the many-body wave-function of the system should vanish when two particles meet. Combined with the regularity of the LLL orbitals, this forces the inclusion of Jastrow factors $(z_i-z_j) ^\ell, \ell\in \N$ in the wave-function, leading to~\eqref{eq:laufunc}. The integer $\ell$ is chosen to accommodate symmetry (the wave-function must be symmetric/antisymmetric under coordinates exchanges $z_i\leftrightarrow z_j$ to describe bosons/fermions) and to fix the filling factor (number of particles per magnetic flux quantum) equal to $1/\ell$. In the absence of an external potential there is then no free variational parameter in the ansatz ~\eqref{eq:laufunc}.

The proposed form \eqref{eq:laufunc} turns out to be sufficiently robust for serving as the basic ingredient of approximate ground states for large classes of repulsive interaction potentials $w$ and  external potentials $V$. 
\medskip

A possible mathematical formulation of Property \textbf{1} is to consider a model pair-potential $w$ whose projection to the $N$-body LLL has~\eqref{eq:laufunc} as an exact zero-energy ground-state~\cite{Haldane-83,TruKiv-85,PapBer-01}. The claim would then be that the next energy level is bounded below, independently of $N$. This problem has remained open so far, see~\cite{LewSei-09,RouSerYng-13b} for further discussion. 
\medskip

Property~\textbf{2} is the main subject of the present contribution. To formulate it properly, we note that the arguments leading to choosing~\eqref{eq:laufunc} as a variational ansatz leave the possibility to choose any $L^2$-normalized function of the form
\begin{equation}\label{eq:fullcorr} 
\Psi_F (z_1,\dots, z_N)= F(z_1,\dots, z_N)\PsiLau (z_1,\dots, z_N)
\end{equation}
with $F$ analytic and symmetric under exchange of the $z_i$.   This form exhausts the class of functions that minimize the magnetic kinetic energy and at the same time avoid repulsive interactions by vanishing at least as $(z_i-z_j)^\ell$ as $z_i$ and $z_j$ come together. In the bosonic case and with $\ell= 2$ these are exactly the ground states of the contact interaction~\cite[Section~2.1]{RouSerYng-13b}. 
The factor $F$ allows an adaption to an external potential. We shall refer to the class of states of the form~\eqref{eq:fullcorr} as {\em fully correlated states}.


Consider now the one-particle density
\begin{equation}\label{eq:density}
\rho_F (z) = N \int_{\R ^{2(N-1)}}  \left| \Psi_F \left(z, z_2,\ldots,z_N \right) \right| ^2 dz_2 \ldots dz_N
\end{equation}
of a state of the form~\eqref{eq:fullcorr}. For $F=1$, i.e., for the Laughlin state~\eqref{eq:laufunc}, it was argued in~\cite{Laughlin-83} and proved rigorously in ~\cite{RouSerYng-13a,RouSerYng-13b} that the density takes essentially the constant value $1/(\pi \ell)$ in the disk $D(0,\sqrt{\pi\ell N})$, and drops quickly to $0$ outside of 
the disk\footnote{This is a consequence of concentration of measure results for $\beta$-ensembles, see~\cite{Serfaty-15,Serfaty-17} for review and references.}. Property~\textbf{2} can then roughly be formulated as 
\begin{equation}\label{eq:property 2}
\boxed{\rho_F \lessapprox \frac{1}{\pi \ell} \mbox{ for any state of the form~\eqref{eq:fullcorr}}.}
\end{equation}
Some care has to be taken in formulating this rigorously for, in view of existing numerical simulations~\cite{Ciftja-06,CifWex-03}, the symbol $\lessapprox$ above cannot actually stand for a microscopic pointwise bound for finite $N$. 

In this paper we improve and generalize results of~\cite{RouYng-14,RouYng-15} (see also~\cite[Chapter~3]{Rougerie-hdr} for further discussion) by proving that~\eqref{eq:property 2} holds in the sense of local averages over length scales $\gg N ^{1/4}$, which are much smaller than the extension of the Laughlin state itself, of order $N ^{1/2}$. In view of recent results~\cite{BauBouNikYau-15,Leble-15b}, which can be applied to the pure Laughlin state, i.e., $F=1$, it is natural to conjecture that~\eqref{eq:property 2} in fact holds on any length scale $N^{\alpha},\alpha >0$. In order to understand the response of the Laughlin state to external potentials, it is {\em essential}, however, to know that the density bound holds for {\em all} analytic factors $F$. The methods of~\cite{BauBouNikYau-15,Leble-15b}, in contrast to those we develop here, do not apply in such generality. 

The bound \eqref{eq:property 2} is an expression of the robustness of the Laughlin state against perturbations by external fields because it forbids a fully-correlated state to accommodate variations of an external potential by concentrating arbitrarily large portions of its mass in energetically favorables places: Any redistribution of the mass must respect the same density  bound \eqref{eq:property 2} as the Laughlin state.

On the other hand it has recently been proved in \cite{RouYng-17} that, if~\eqref{eq:property 2} is known to hold for all strongly correlated states, then the minimal energy in an external potential within this class of states can asymptotically be attained in states saturating the bound~\eqref{eq:property 2} and having the simple form
\begin{equation}\label{eq:lesscorr}
\Psi_f (z_1,\ldots,z_N) = \prod_{j= 1} ^N f (z_j) \PsiLau (z_1,\ldots,z_N) 
\end{equation}
where $f$ is a polynomial of a single variable.  Physically, this means that it is never favorable to add more correlations to the Laughlin state in order to accommodate an external potential, provided the latter is weak enough so that it does not make the ground state jump across the energy gap assumed in Property~\textbf{1}.

The zeros $a_j\in\C$ of the polynomial $f$, 
\begin{equation}\label{eq:quasi holes}
 f (z) = {\rm const.} \prod_{j=1} ^J (z-a_j),
\end{equation}
are interpreted as the locations of {\em quasi-holes}, each carrying a charge $1/\ell$.
If the locations $a_1,\ldots,a_J$ are treated quantum-mechanically instead of classically as above, one should expect that the quasi-holes behave as anyons~\cite{AroSchWil-84,LunRou-16,ZhaSreGemJai-14} with statistics parameter $-1/\ell$.  

The results of the present paper together with \cite{RouYng-17} confirm rigorously 
Laughlin's original intuition~\cite{Laughlin-83,Laughlin-87} that the response of the Laughlin state to a disorder potential is  to generate uncorrelated quasi-holes, which is also in accord with experimental studies~\cite{YacobiEtal-04}.

\medskip

In the next section we state our main results precisely and discuss them further (they were announced in the short paper~\cite{LieRouYng-16}). The rest of the paper is devoted to their proofs.

\bigskip

\noindent\textbf{Acknowledgments:} We received financial support from the French ANR project ANR-13-JS01-0005-01 (N.~Rougerie) and the US NSF grant PHY-1265118  (E.~H.~Lieb). 

\section{Main results}\label{sec:statements}

\subsection{Universal density bounds}\label{sec:results}

As discussed in the introduction (see also~\cite{RouYng-14,RouYng-15}), we assume that the magnetic field and repulsive interactions are strong enough, so that the ground state of the system has the form~\eqref{eq:fullcorr}. We aim at proving that there is a universal upper bound to the local density of all such states, on scales much smaller than the total size of the system. Thus, let
\begin{equation}\label{laufunc}
\PsiLau (z_1,\ldots,z_N) =\cLau \prod_{i<j}(z_i-z_j)^{\ell} e^{-\sum_{i=1}^N|z_i|^2/2}
\end{equation}
be the $L^2$-normalized Laughlin function with exponent $\ell \in \N$, in units where the magnetic length is $1/\sqrt{2}$.  As previously explained  we consider analytic perturbations thereof:
\begin{equation}\label{eq:fullcorr set}
\cL^N _{\ell} = \left\{ \Psi_F = F(z_1,\dots, z_N)\PsiLau \:\big| \: F \mbox{ analytic and symmetric}, \: \Vert \Psi_F \Vert_{L^2 (\R ^{2N})} = 1 \right\}.
\end{equation}
For a function $\Psi_F \in \cL^N _{\ell}$ we denote by
\begin{equation}\label{eq:original density}
\rho_F (z) = N \int_{\R ^{2(N-1)}}  \left| \Psi_F \left(z, z_2,\ldots,z_N \right) \right| ^2 dz_2 \ldots dz_N
\end{equation}
the associated one-particle density. Our main result is a universal upper bound on $\rho_F$, holding on mesoscopic length scales: 

\begin{theorem}[\textbf{Rigidity bound for fully-correlated states}]\label{thm:main}\mbox{}\\ 
For any $\alpha > 1/4$, any disk  $D$ of radius $N^{\alpha}$ and any (sequence of) states $\Psi_F\in \cL_\ell^N$ we have 
\begin{equation}\label{bound}
\int_{D} \rho_F \leq \frac{1}{\pi\ell} |D| (1+o(1))
\end{equation}
where $|D|$ is the area of the disk and $o(1)$ tends to zero as $N\to \infty$. 

More generally, for any open set with Lipschitz boundary $\Omega$, denote by $\Omega_r$ the set obtained by dilating $\Omega$ around some origin by a factor $r = N ^\alpha$. Then, for any $\alpha > 1/4$ 
\begin{equation}\label{eq:gen dens bound}
\int_{\Omega_r} \rho_F \leq \frac{1}{\pi\ell} |\Omega_r| (1+o(1))
\end{equation}
as $N\to \infty.$
%
%
\end{theorem}

\subsection{Potential energy estimates}

With Theorem~\ref{thm:main} in hand, we can return to the original physical problem~\eqref{eq:mag hamil} in the presence of an external potential $V$. If the magnetic kinetic energy is frozen by restricting one-body orbitals to the LLL, and if we further assume that the interaction energy is also frozen by restricting to~\eqref{eq:fullcorr set}, the only non-trivial energy term left in the Hamiltonian is set by $V$. An interesting problem consists in studying the minimization of this energy amongst fully-correlated states:
\begin{equation}\label{eq:pot ener}
E_V (N,\ell) := \inf \left\{ \intR V (z) \rhoF (z) dz, \: \rhoF \mbox{ as in~\eqref{eq:original density}, }  \Psi_F \in \cL^N _{\ell} \right\}. 
\end{equation}
This infimum will not necessarily be  attained, but it is easy to see that, for fixed $N$, the density $\rhoF$ is bounded in $L^{\infty}$ uniformly in $F$ (this follows from the fact that all one-body orbitals are in the LLL~\cite{Carlen-91}, see the discussion in~\cite[Remark~2.2]{RouYng-14}). Therefore the infimum exists for any $V$ such that its negative part, $V_-$, is in  $L ^1 (\R ^2)$. 

In view of Theorem~\ref{thm:main}, it is fairly natural to expect a lower bound to~\eqref{eq:pot ener} in terms of the ``bathtub energy'', defined as~\cite[Theorem~1.14]{LieLos-01}  
\begin{equation}\label{eq:bathtub}
E^{\rm bt}_V (\ell) := \min\left\{ \intR V \rho \: \Big| \: 0 \leq \rho \leq \frac{1}{\pi \ell},\ \intR\rho = N \right\}.
\end{equation}
We can prove this for a large class of external potentials $V$, that we now describe. Since the typical extension of the Laughlin state is of order $\sqrt{N}$, it is more convenient to state our assumptions in terms of a scaled version of the potential.

\begin{asumption}[\textbf{The external potential}]\label{asum:pot}\mbox{}\\
Let $V \in C^2 (\R ^2)$ be a (sequence of) potential(s) that we write in the manner 
\begin{equation}\label{eq:scale pot asum}
V (x) = U \left( \frac{x}{\sqrt{N}}\right) 
\end{equation}
for a (possibly $N$-dependent) $U\in C^2 (\R ^2)$. Assume that there exists a fixed radius $R >0$ and an $\alpha > 0$ such that
\begin{enumerate}
 \item \textbf{\emph{Local behavior:}} 
\begin{equation}\label{eq:pot vari}
\norm{\nabla U}_{L ^{\infty} (D(0,R))} \leq C N ^{1/2-\alpha}, \quad  \norm{\Delta U}_{L ^{\infty} (D(0,R))} \leq C N ^{1- 2\alpha}.
\end{equation}
 \item \textbf{\emph{Growth at infinity:}} $U \geq \underline{U}$ outside of $D(0,R)$ for a potential $\underline{U}\in C^2 (\R ^2)$ satisfying
\begin{equation}\label{eq:pot infi 1} 
\norm{\nabla \underline{U}}_{L ^{\infty} (\R^2)} \leq C, \quad  \norm{\Delta \underline{U}}_{L ^{\infty} (\R^2)} \leq C
\end{equation}
and 
\begin{equation}\label{eq:pot infi 2}
\int_{\R ^2} \exp(- \underline{U}) \leq C  
\end{equation}
independently of $N$.
\end{enumerate}
\end{asumption}

A \emph{typical example} we have in mind is a potential $V$ of the sort 
\begin{equation}\label{eq:typical pot}
V (x) = V_{\rm trap} (x) + V_{\rm disorder} (x)=  U_{\rm trap} \left( \frac{x}{\sqrt{N}}\right) + U_{\rm disorder} \left( \frac{x}{\sqrt{N}}\right)
\end{equation}
where the part $U_{\rm trap}$ represents trapping in a given sample, and satisfies Item 2 of the assumptions, whereas $U_{\rm disorder}$ represents some disorder due to impurities in the sample, is compactly supported in a fixed disk and satisfies Item 1 of the assumption.

Item 1 means that we consider a potential $V$ whose variations happen on length scales $O(N^{\alpha}), \alpha >0$, the best scale we may hope for (but cannot quite reach) in our density bounds, Theorem~\ref{thm:main}. Note that this should \emph{not} be interpreted as saying that the potential energy bounds we are about to state work on the optimal length scale, see Remark~\ref{rem:lengths} below. Item 2 is essentially a mild trapping assumption. It is clearly satisfied by any regular enough function $U$ growing at least like $c\log|x|, c > 2$ at infinity. 

We shall prove the following corollary, which improves~\cite[Theorem~1.1]{RouYng-15}:

\begin{corollary}[\textbf{Lower bound to the potential energy}]\label{cor:pot ener}\mbox{}\\
Let $V$ be as in Assumption~\ref{asum:pot}. Then 
\begin{equation}\label{eq:bathtub low bound}
E_{V} (N,\ell) \geq E^{\rm bt}_{V} (\ell) (1+o(1)) + o(N)
\end{equation}
as $N\to \infty$.
\end{corollary}

\begin{remark}[Affordable length scales]\label{rem:lengths}\mbox{}\\
In accordance with the limitations of Theorem~\ref{thm:main}, the above result is meaningful when $\alpha > 1/4$ in Assumption~\ref{asum:pot}, i.e., the potential varies on length scales much larger than $N^{1/4}$. To appreciate this, recall~\cite[Theorem~1.14]{LieLos-01} that the minimizer of the bathtub energy takes only the values $0$ and $(\pi \ell)^{-1}$ almost everywhere. Under Assumption~\ref{asum:pot}, the typical length scale of the potential $V$ is $L=N^{\alpha}$. It can have a variation of order $rL$ in a disk of radius $r$. For comparison (and without loss of generality) let us set the reference energy to be the minimum of $V$ in the disk. The difference in bathtub energy according to whether the minimizer is, say, equal to $0$ or $(\pi \ell)^{-1}$ everywhere in said disk cannot be larger than $O(r ^3 L)$ (area times maximal value of the potential). Pick now $r= L$. For the energy difference to dominate the error term $o(N)$ in~\eqref{eq:bathtub low bound}, we need $L^4 \gg N$, that is $N^{4\alpha} \gg N$, and this is guaranteed only if $\alpha > 1/4$. The bound~\eqref{eq:bathtub low bound} thus captures the behavior of the bathtub minimizer on its smallest length scale (i.e., that of $V$) only if the latter is $\gg N ^{1/4}$. \hfill  $\diamond$
\end{remark}

In the recent paper ~\cite{RouYng-17}, Theorem 2.3, an upper energy bound matching the lower bound \eqref{eq:bathtub low bound} has been proved:

\begin{proposition}[\textbf{Upper bound to the potential energy, \cite{RouYng-17}}]\label{pro:up bound}\mbox{}\\
Under the additional\footnote{Mostly technical, see~\cite{RouYng-17} for details.} condition that $U$ in Assumption~\ref{asum:pot} is fixed ($N$-independent), we have
\begin{equation}\label{eq:bathtub upper bound}
E_{V} (N,\ell) \leq E^{\rm bt}_{V} (\ell) (1+O(N^{-1/4})).
\end{equation}
\end{proposition}

The above is obtained by using trial states of the form~\eqref{eq:lesscorr}. For the special case where $V$ is radially symmetric and monotonously increasing in the radial variable, or is of a Mexican-hat form, such an upper bound was proved earlier in ~\cite{RouYng-14,RouYng-15}.
A remarkable consequence is that the Laughlin state stays an approximate ground state in any radial increasing potential, however steep. We refer to~\cite{RouYng-14,RouYng-15} again and to the companion paper~\cite{LieRouYng-16} for further discussion of the significance of this. 


\subsection{Proof outline: screening properties in the plasma analogy}\label{sec:proof out}

The backbone of the proof of Theorem \ref{thm:main} is Laughlin's plasma analogy, a very fruitful mapping of the many-body density of $\PsiLau$, and more generally $\Psi_F$,  to a problem in classical statistical mechanics, originating in~\cite{Laughlin-83,Laughlin-87}. As in our previous papers~\cite{RouSerYng-13a,RouSerYng-13b,RouYng-14,RouYng-15} we generalize this to any fully-correlated state~\eqref{eq:fullcorr}. We rewrite the $N$-particle probability density of any fully correlated state $\Psi_N$ in the form of a Boltzmann-Gibbs factor 
\begin{equation}\label{eq:plasma analogy 1}
\left| \Psi_F (z_1,\ldots,z_N) \right|^2 = \frac{1}{\cZ_N} \exp\left( - H_N (z_1,\ldots,z_N) \right) 
\end{equation}
with a classical Hamiltonian
\begin{equation}\label{eq:class Hamil 1}
H_N (z_1,\ldots,z_N) = \sum_{j=1} ^N |z_j| ^2 - 2 \ell \sum_{1\leq i<j\leq N} \log |z_i-z_j| - 2 \log |F (z_1,\ldots,z_N)|. 
\end{equation}
To interpret~\eqref{eq:class Hamil 1} as an energy functional for $N$ classical positively charged particles with coordinates $z_1,\ldots,z_N\in \C \simeq \R^2$ we recall that Poisson's equation relates the electrostatic  potential $\varphi$ generated by a charge distribution $\rho$ in the manner 
\begin{equation}\label{eq:Poisson}
-\Delta \varphi = 2\pi \rho. 
\end{equation}
where $\Delta$ is the two-dimensional Laplacian. Thus one can see~\eqref{eq:class Hamil 1} as the energy of $N$ positively charged particles subject to
\begin{itemize}
 \item the external potential generated by a constant background of opposite charge. Since 
 \begin{equation}\label{eq:background}
 -\Delta_{z_j} |z_j| ^2 = - 4  
 \end{equation}
the potential $|z_j| ^2$ can be interpreted as being generated by a uniform charge background of negative density $- 2/\pi$. 
\item pairwise Coulomb repulsion:
\begin{equation}\label{eq:Coulomb rep} 
-\Delta_{z_j} \left( - 2\ell \sum_{1\leq i<j\leq N} \log |z_i-z_j| \right) = \sum_{_i \neq j} 4 \pi \ell \delta_{z_i=z_j}, 
\end{equation}
thus the second term in~\eqref{eq:class Hamil 1} is the 2D Coulomb interaction between particles of identical charge $\sqrt{2\ell}$.
\item the potential generated by an essentially arbitrary, but \emph{positive}, (many particle) charge distribution encoded in the (modulus of the) analytic function $F$. Indeed, for any $j=1\ldots N$,
\begin{equation}\label{eq:subharm 1}
- \Delta_{z_j} \left( - \log |F(z_1,\ldots,z_N)| \right) \geq 0.
\end{equation}
\end{itemize}
This is, in fact,  the only property of $F$ we use in our method.  Equation~\eqref{eq:subharm 1} means that $-\log |F|$ is  superharmonic in every variable, because $F$ is analytic \footnote{One can write it in Weierstrass factorized form to see that~\eqref{eq:subharm 1} holds, or use the Cauchy integral formula for the value of an analytic function at the center of any disk and rely on the characterization of superharmonic functions in terms of averages in disks~\cite[Chapter~9]{LieLos-01}.}. 

\medskip

A plausibility argument for our main results goes as follows:
\begin{enumerate}
 \item The effective temperature in~\eqref{eq:plasma analogy 1} is small compared to the total energy, so that the Boltzmann-Gibbs factor is essentially concentrated around minimizing configurations of the Hamiltonian~\eqref{eq:class Hamil 1}.
 \item For $F= 1$, i.e. in the case of the pure Laughlin state, one should expect that the local density of points in a minimizing configuration is close to $1/(\ell\pi)$.
 Indeed, this is the condition for the negatively charged fixed constant background (c.f.~\eqref{eq:background}) to efficiently screen the positively charged mobile particles $z_j,j= 1\ldots N$ (c.f.~\eqref{eq:Coulomb rep}).
 \item If $F\neq 1$ we add an additional positive charge distribution to the game because of~\eqref{eq:subharm 1}. Since the mobile particles $z_j,j=1 \ldots N$ are positively charged in this representation, one can expect that such an addition will exert a repelling force on the density distribution of the mobile particles and reduce it, at least on the average. 
 \end{enumerate}
 
This intuitive picture is, however, not as clear as it may look at first sight. If we imagine the potential term associated to $F$ as being generated by point-like positively charged particles, it is not obvious {\it a priori}  why  several such charges cannot work together and conspire to increase the density {\em locally}.  Arranged tightly on a circle they could, perhaps, concentrate the density at a high value around the center of the circle. Note also that $- 2 \log |F (z_1,\ldots,z_N)| $ is generated by a \emph{mobile} charge distribution because $\Delta_{z_j} \log |F(z_1,\ldots,z_N)|$ is in general a function of $z_1,\ldots,z_N$ and only constant in special cases (essentially when $F = f ^{\otimes N}$ is a tensor power, as in~\eqref{eq:lesscorr}). Thus, the positions of the additional charges are in general correlated with those of the actual charges. Moreover, as the example $F(z_1,\dots,z_N)=\prod_i\exp(cz_i)$ shows, the area occupied by the bulk of the density  can be stretched in one direction but compressed in another. 

The heuristic argument (1)-(3) above is, in any case, certainly far from a rigorous proof of the bound \eqref{bound}. To obtain such a proof, the reduction to considering only ground states of~\eqref{eq:class Hamil 1} (rather than the full Gibbs state) presents, perhaps,  the least difficulty. To appreciate this, it is helpful to rescale lengths: we can write 
\begin{equation}\label{eq:scale Gibbs}
\mu_F (z_1,\ldots,z_N) := N ^ N \left| \Psi_F ( \sqrt {N} z_1,\ldots, \sqrt{N} z_N) \right|^2 = \frac{1}{\tilde{\mathcal{Z}}_N} \exp\left( -N  \tilde{H}_N (z_1,\ldots,z_ N) \right)
\end{equation}
as a probability measure with 
\begin{equation}\label{eq:scale Hamil}
\tilde{H}_N (z_1,\ldots,z_N) = \sum_{j=1} ^N |z_j| ^2 - 2 \frac{\ell}{N} \sum_{1\leq i<j\leq N} \log |z_i-z_j| - \frac{2}{N} \log |F ( \sqrt{N} z_1,\ldots, \sqrt{N}z_N)|. 
\end{equation}
Thus, modulo a simple change of length-scales, we can reduce the aim to a statistical mechanics problem 
\begin{itemize}
\item in mean-field scaling, the first and second sum in the above 
being of the same order for large $N$.
\item with small effective temperature $T = N ^{-1}$. 
\end{itemize}
This makes it clear why mostly ground-state configurations are relevant for the proof, although some care is needed to estimate the effect of entropy terms in the classical free energy. In fact, the reduction to classical ground state configurations is responsible for the fact that our main result holds only on length scales $\gg N ^{1/4}$. For ground-states instead of Gibbs states, the result holds on scales $\gg 1$. There is thus certainly room for improvement in this particular part of the proof, see Remark~\ref{rem:rel res} below. 

With the reduction to ground states of~\eqref{eq:class Hamil 1} out of the way, the main question is now to understand points $(2)$ and $(3)$ of the list above. In our approach, the solution is essentially based on \emph{screening} properties of ground states of classical Coulomb systems, namely their tendency to achieve local neutrality in order for the total potential generated by the positively and negatively charged particles to compensate one another. 

Consider a minimizing configuration $Z_N ^0 = \left\{ z_1 ^0,\ldots,z_N ^0 \right\}$ for the classical Hamiltonian~\eqref{eq:class Hamil 1} (the labeling of the points is irrelevant by symmetry of~\eqref{eq:class Hamil 1} under permutations). To any subset $\left\{z_1^0,\ldots,z_K ^0\right\}$ of the configuration (in fact, any $K$-tuple of points in~$\R ^2$) we associate an open \emph{screening region} $\Sigma_K$ such that 
\begin{equation}\label{eq:screen intro}
\Phi (x) = \int_{\Sigma_K} \frac{1}{\pi \ell} \log |x-y| dy - \sum_{j=1} ^K \log |x-z_k| \begin{cases} = 0 \mbox{ outside } \Sigma_K \\
                                                                     > 0 \mbox{ inside } \Sigma_K.
                                                                    \end{cases}
\end{equation}
The above means that the total potential generated by the $K$ positive charges under consideration and a constant background of opposite charge density $(\pi \ell)^{-1}$ contained in $\Sigma_K$ vanishes identically outside $\Sigma_K$. The background density coming from the first term in~\eqref{eq:class Hamil 1} thus completely screens the $K$ point charges outside of the set $\Sigma_K$. Note that for~\eqref{eq:screen intro} to hold we must have neutrality, 
\begin{equation}\label{eq:intro neutral}
|\Sigma_K| = \pi \ell K. 
\end{equation}
The sets $\Sigma_K$ are constructed in our approach via the minimization of an auxiliary functional bearing some formal resemblance to a 2D version of the usual Thomas-Fermi theory for molecules~\cite{Lieb-81b}. Accepting their existence, the rest of the proof has two main steps, going roughly as follows
\begin{itemize}
 \item Given the screening region $\Sigma_K$ associated to a subset $\left\{z_1^0,\ldots,z_K ^0\right\}$ of the minimizing configuration, no other point of a minimizing configuration of particles can lie within $\Sigma_K$. We refer to this as the \emph{exclusion property}. It holds for ground states of~\eqref{eq:class Hamil 1}, as we shall prove, because the third term herein is superharmonic. 
 \item \emph{Any} configuration of points satisfying the exclusion property must have its local density bounded above by $(\pi \ell)^{-1}$. If some region contains a larger density, consider a disk included in it. The screening region associated to the points contained in the disk would have, because of~\eqref{eq:intro neutral}, to leak outside of the disk. It would then overlap with other points of the configuration, which is impossible because of the exclusion property.  
\end{itemize}

The construction of the screening set, via a variant of Thomas-Fermi theory, is done in Section~\ref{sec:TF}. Some care is needed in handling the auxiliary variational problem, because it is of a non-standard type. We next prove the local density bound for ground states of~\eqref{eq:class Hamil 1} in Section~\ref{sec:class GS}, following the above strategy. The application to low-temperature Gibbs states~\eqref{eq:plasma analogy 1} occupies Section~\ref{sec:dens Gibbs}, where we prove Theorem~\ref{thm:main}. Finally we deduce Corollary~\ref{cor:pot ener} in Section~\ref{sec:ener low bound}. 

\begin{remark}[Recent related results]\label{rem:rel res}\mbox{}\\
The classical Hamiltonian~\eqref{eq:class Hamil 1} and the associated Gibbs state~\eqref{eq:plasma analogy 1} in the pure Laughlin case $F=1$ have attracted much attention in their own right. Apart from the connection with the FQHE we are primarily concerned about here, one can see them as basic models for trapped plasmas. In this context, a generalization often considered (2D log-gas or 2D $\beta$-ensemble) is to replace $|z_j|^2 \to V (z_j)$ in~\eqref{eq:scale Hamil}, with $V : \R ^2 \mapsto \R$ a general trapping potential. Minimizing configurations are then weighted Fekete sets~\cite{SafTot-97}, and Gibbs states for special values of the temperature are connected to certain random matrices ensembles~\cite{AndGuiZei-10,Forrester-10}. 

For general $V$ the density is not flat: the limiting density of points is proportional to $\Delta V$ on its support. One expects the true minimizers and low-temperature Gibbs states to follow this profile on any microscopic length scale $\gg N ^{-1/2}$ (in the scaled variables used in~\eqref{eq:scale Hamil}). This has recently been proved in~\cite{AmeOrt-12,RotSer-14} for ground states (see also~\cite{PetRot-16} for higher dimensional Coulomb and Riesz gases) and in~\cite{BauBouNikYau-15,BauBouNikYau-16,Leble-15b,LebSer-16} for Gibbs states. See also earlier partial results (not on the optimal scale) in~\cite{SanSer-14,RouSer-14,PetSer-14} and~\cite{Ameur-16,Ameur-17,BetSan-14,ChaGozZit-13,ChaHarMai-16,LebSerZeiWu-15,LebSer-15,SanSer-14a} for related recent literature.

One can see our main result as a generalization of the fine-scale rigidity estimates of~\cite{AmeOrt-12,RotSer-14,BauBouNikYau-15,Leble-15b} to the case of a non-trivial $F$. Motivated by FQHE physics, {\em this generalization is our chief concern}. For general analytic $F$ one should expect only local density upper bounds, but no corresponding lower bound in general. None of the available methods seem to apply here, but the new approach we sketched above might be of interest in the usual context with $F = 1$ and a general $V$ in~\eqref{eq:scale Hamil}. Note that we obtain the expected density upper bound on essentially optimal length scales for minimizing configurations (see Proposition~\ref{pro:dens plasma} below), but that it remains an open problem to do the same for Gibbs states, i.e. allow any $\alpha >0$ instead of only $\alpha > 1/4$ in Theorem~\ref{thm:main}. \hfill$\diamond$
\end{remark}

\section{Incompressible 2D Thomas-Fermi molecules}\label{sec:TF}

In this section we discuss a Thomas-Fermi-like variational theory of 2D molecules. We consider the functional 
\begin{equation}\label{eq:TF func} 
\mathcal \ETF [\sigma]=-\int_{\R^2} \Vnuc(x)\sigma(x)\,dx + D(\sigma, \sigma)
\end{equation}
where $\sigma$ is the density distribution of negatively charged \lq\lq electrons\rq\rq, 
\begin{equation}\label{eq:nucpot} 
\Vnuc(x)=-\sum_{i=1}^K \log|x-x_i| 
\end{equation}
is the 2D electrostatic potential generated by $K$ point \lq\lq nuclei\rq\rq\ of charge $+1$ located at positions $x_1,\ldots,x_K$ in the plane
and 
\begin{equation}\label{eq:coul ener}
 D(\sigma,\sigma)=-\half \iint_{\R ^2 \times \R ^2} \sigma(x)\log|x-y|\sigma(y)\, dx\, dy 
\end{equation}
is the electrostatic self-energy of the \lq\lq electrons\rq\rq. 

We emphasize that in the application we have in mind, the designations \lq\lq nuclei\rq\rq and ``electrons\rq\rq\, do not refer to real physical entities. The terminology is just chosen for convenience and analogy with standard TF theory for molecules~\cite{CatBriLio-98,Lieb-81b}. 

We consider the minimization of $\ETF$
over the variational set 
\begin{equation}\label{eq:MTF}
\TFM := \left\{\sigma\in L^\infty(\mathbb R^2)\cap  L^1\left(\mathbb R^2, \log(2+|x|)dx\right), 0\leq \sigma \leq 1, \, \int_{\R ^2}\sigma = K\right\}
\end{equation}
which means that we are interested in an ``incompressible neutral molecule''. The term ``incompressible" refers to the constraint $\sigma \leq 1$ on the density of electrons and  ``neutral" to the requirement that the total charge of the electrons equals that of the nuclei.\medskip

\begin{remark}[Related models]\label{rem:models}\mbox{}\\
 From a  physics point of view, a more natural functional to consider would be
\beq
\mathcal E_\tau [\sigma]= \int_{\R^2} \tau(\sigma) -\int_{\R^2} \Vnuc(x)\sigma(x)\,dx + D(\sigma, \sigma)
\eeq
with a function $\tau$ representing the kinetic energy density of the electrons. For $ \tau(\sigma)= {\rm (const.)}\cdot\sigma^2$ this term is the standard  local density/semi-classical approximation of the quantum kinetic energy of a 2D electron gas. The minimization problem \eqref{eq:TF func}--\eqref{eq:MTF}, which is the relevant one for our purpose,  corresponds formally to taking $\tau(\sigma)=\sigma^p$ with $p\to\infty$ to enforce the uniform upper bound $\sigma\leq 1$, but we shall not follow this $p\to \infty$ route. 

Many results we prove below have a natural extension to the case where $\Vnuc$ is a more general potential, as long as it is generated by a compactly supported charge distribution. One could also consider corresponding 3D models as e.g. in~\cite{BonKnuRog-16} where results parallel to those of this section have been obtained. Note that the 3D setting is in fact simpler to deal with, because of the fall-off at infinity of the 3D Coulomb kernel. 

We also mention that there is a kinship with recently studied flocking/swarming models~\cite{BurChoTop-15,FraLie-16} where constrained minimization problems similar to the above have been considered and proved to support phase transitions.  

\hfill $\diamond$
\end{remark}

The facts we shall need about the variational problem for $\ETF$
 are collected in the following theorem. 
 
\begin{theorem}[\textbf{Incompressible neutral Thomas-Fermi molecules}]\label{thm:TF theory}\mbox{}\\
Let $K$ be a positive integer. We have 
\begin{enumerate}
 \item \textbf{Well-posedness.} There exists a  unique $\TFmin$ that minimizes $\ETF$ over the variational set $\TFM$. 
 \item \textbf{Potential and variational inequalities.} Let 
 \begin{equation}\label{eq:TF pot}
  \TFpot := \Vnuc + \TFmin * \log |\,.\,| 
 \end{equation}
be the total electrostatic potential of the molecule. It is continuous and once continuously differentiable away from the nuclei and tends to zero at infinity. Moreover, for almost all $x$,
\begin{align}\label{eq:TF eq}
\TFpot(x) &> 0 \mbox{ if }  \TFmin(x)>0 \nonumber \\
\TFpot(x) &= 0 \mbox{ if } \TFmin(x)=0.
\end{align}
\item \textbf{Values of the electron density.}
 The set $\left\{ x\,:\, 0< \TFmin(x) <1 \right\}$ has zero Lebesgue measure. Hence  $\TFmin$ takes only the values 1 or 0 a.e.
 \item \textbf{Support of the electron density.}\label{eq:suppTFmin} \medskip
 
 \noindent (i) The set  $\{x:\, \TFmin(x)>0\}$ is essentially bounded:
If $D(0,r)$ denotes a disk with radius $r$ and center at 0, then for any $R>\max |x_i|$ we have, up to a set of measure zero,  \begin{equation}\label{eq:support TF}
\{x: \TFmin(x)>0\} \subset D\left(0, R + \sqrt{M_R}\right)
\end{equation}
with
 \beq\label{eq:MR} M_R := \frac{1}{\pi}\sup_{|x| = R} \TFpot (x).\eeq

 \noindent (ii) Up to a set of measure zero, $\{x: \TFmin(x)>0\}$ is equal to the open set
 \beq\label{eq:sigma TF} \Sigma^{\rm TF}(x_1,\dots,x_K)=\{x\,:\, \TFpot(x)>0\}\eeq
and the area of $\Sigma^{\rm TF}(x_1,\dots,x_K)$ is equal to $K$.
 
  \noindent (iii) The locations $x_1,\ldots,x_K$ of the nuclei lie within  $\Sigma^{\rm TF}(x_1,\dots,x_K)$ and all connected components of the latter set contain at least one $x_j$. 
 
 \noindent (iv) 
\beq \label{eq:inclusion screen set}
\Sigma^{\rm TF}(x_1,\dots,x_{K-1})\subset \Sigma^{\rm TF}(x_1,\dots,x_K).
\eeq

 \noindent (v) For a single nucleus, 
 \beq\STF(x_1)=\{x:\, |x-x_1|<1/\sqrt \pi\}.\eeq
\end{enumerate}
\end{theorem}
\bigskip

The inclusion in~\eqref{eq:inclusion screen set} says that when a nucleus is added, the density $\TFmin$ and the potential $\TFpot$ increase everywhere. The rest of this section is devoted to the proof of these results. 

\subsection{Existence and uniqueness of a minimizer}\label{sec:ex un}

The uniqueness is simpler than the existence and we consider it first:
\begin{lemma}[\textbf{Strict convexity and uniqueness}]\label{lem:convexity}\mbox{}\\
The functional $\ETF$ is strictly convex on $\TFM$. Consequently, the minimizer $\TFmin$, if it exists, is unique.
\end{lemma}

\begin{proof}
We only have to show that $D(\sigma,\sigma)$ is strictly convex on $\TFM$. This is a consequence of the positivity of $D(\mu,\mu)$ for neutral charge distributions $\mu$ as in~\cite[Lemma~3.2]{RouSerYng-13b}. The general fact is that
\begin{equation}
 D(\mu,\mu) = \frac{1}{2\pi} \int_{\R ^2} \left( \int_{\R ^2} \frac{1}{|t-z|} d\mu(z) \right) ^2 dt
\end{equation}
whenever $\int_{\R^2} \mu = 0$, see~\cite[Chapter I, Lemma 1.8]{SafTot-97}. This implies the claimed strict convexity property: pick two measures $\mu_1$ and $\mu_2$ with $\int\mu_1=\int\mu_2$. Then 
\beq
\frac{1}{2} D(\mu_1,\mu_1) +\frac{1}{2} D(\mu_2,\mu_2) - D\left(\frac{1}{2}\mu_1+ \frac{1}{2}\mu_2, \frac{1}{2}\mu_1+ \frac{1}{2}\mu_2\right) = \frac{1}{4} D\left( \mu_1- \mu_2, \mu_1 - \mu_2\right) \geq 0
\eeq
since $\int_{\R ^2} \mu_1 = \int_{\R ^2} \mu_2 = 1$. Equality holds if and only if $\mu_1=\mu_2$ a.e. Uniqueness of the minimizer follows because we minimize a strictly convex functional on a convex set.
\end{proof}

Next we consider the boundedness from below and lower semicontinuity of the functional~$\ETF$ .

\begin{lemma}[\textbf{Boundedness from below}]\label{lem:low bound}\mbox{}\\
The functional $\ETF$ is uniformly bounded from below on $\TFM$. It is moreover lower semi-continuous under strong convergence in $L^1 \cap L ^{p}$ for any $p>1$.
\end{lemma}

\begin{proof}
For any $\sigma \in \TFM$ we can use the normalization to write 
\begin{equation}\label{eq:rewrite ener TF}
 \ETF[\sigma] = \frac{1}{2} \iint_{\R^2 \times \R ^2} \left( K ^{-1} \sum_{i=1} ^K \log |x-x_i| + K^{-1}\sum_{i=1} ^K \log |y-x_i| - \log|x-y| \right) \sigma (x) \sigma (y)dxdy. 
\end{equation}
Then, since $|x-y| \leq (1+|x|)(1+|y|)$, we have 
\beq \ETF[\sigma] \geq K \int_{\R^2} \left( K ^{-1} \sum_{i=1} ^K \log |x-x_i| - \log(1+|x|) \right) \sigma (x) dx.\eeq
The function in parenthesis is clearly uniformly bounded from below outside of a disk $D(0,R)$ that contains all the nuclei. On the disk $D(0,R)$ it is locally integrable. Since $\sigma\in\TFM$ is uniformly bounded in $\in L^1 (\R ^2) \cap L ^{\infty} (\R ^2)$ and non-negative, we can split the integration domain according to whether $|x| \leq R$ or the other way around, and conclude that both pieces are bounded from below uniformly.

The lower semi-continuity follows from the same kind of considerations, using in addition Fatou's lemma and the fact that $\log |\,.\,| \in L^q _{\rm loc} (\R ^2)$ for any $1 \leq q < \infty$. 
\end{proof}

We now prove  the existence of a minimizer. Note that if we had been minimizing under the constraint that the mass be equal to $N<K$, existence of a minimizer would follow from standard arguments discussed e.g. in~\cite{SafTot-97}. At neutrality, when the number of electrons equals that of nuclei, this is slightly more subtle and we use some tricks we learned from~\cite{BetSan-14,Hardy-12,HarKui-12}.

\begin{proof}[Proof of Theorem~\ref{thm:TF theory}, Item~1]
Uniqueness of the minimizer is already contained in Lemma~\ref{lem:convexity}. For the existence, consider a minimizing sequence $(\sigma_n)_{n\in \N}$. In view of the definition of $\TFM$ we can extract a (not relabeled) subsequence which converges weakly-$\star$ as a Radon measure, weakly in any $L^p (\R ^2)$ for any $1<p<\infty$, and weakly-$\star$ in $L^{\infty} (\R ^2)$ to a candidate minimizer $\TFmin$, i.e.
\begin{equation}\label{eq:weak CV}
\int_{\R ^2} f \sigma_n \to \int_{\R ^2} f \sigma 
\end{equation}
for any function $f \in L^{q} (\R^2), 1 \leq q < \infty$ (the convergence as a Radon measure corresponds to $f$ continuous with compact support). Using Mazur's lemma~\cite[Theorem~2.13]{LieLos-01} and the convexity of the functional we could assume without loss that the convergence is strong in any $L^{q} (\R^2), 1 < q < \infty$, but we shall not use this.

By the lower semicontinuity proved in Lemma~\ref{lem:low bound} we obtain
\beq\ETF [\TFmin]\leq \inf\{\ETF[\sigma]:\, \sigma\in \mathcal M^{\rm TF}\}.\eeq
Hence we only have to prove that $\TFmin \in \TFM$. The upper and lower constraints on $\sigma_n$ pass to the limit easily, but it remains to show that
\beq \int_{\R^2} \TFmin = K.\eeq
This follows from the fact that the sequence $(\sigma_n)_{n\in \N}$ is tight, i.e., no mass escapes to infinity in the limit. To prove this we pick some $R>0$ and rewrite the energy as in~\eqref{eq:rewrite ener TF}, then split the part of the double integral where $|x| \geq R, |y|\geq R$  from the rest: 
\begin{multline}
 \ETF[\sigma_n] 
 \\= \frac{1}{2} \iint_{|x|\geq R, |y|\geq R} \left( K ^{-1} \sum_{i=1} ^K \log |x-x_i| + K^{-1}\sum_{i=1} ^K \log |y-x_i| - \log|x-y| \right) \sigma_n (x) \sigma_n (y)dxdy \\
 + \frac{1}{2} \iint_{|x|\leq R \mbox{ or } |y|\leq R} \left( K ^{-1} \sum_{i=1} ^K \log |x-x_i| + K^{-1}\sum_{i=1} ^K \log |y-x_i| - \log|x-y| \right) \sigma_n (x) \sigma_n (y)dxdy.
\end{multline}
The second term is bounded below independently of $n$ and $R$ by the same arguments as in the proof of Lemma~\ref{lem:low bound}. For the rest we rewrite 
\begin{multline}
 K ^{-1} \sum_{i=1} ^K \log |x-x_i| + K^{-1}\sum_{i=1} ^K \log |y-x_i| - \log|x-y|  \\
 = \left( K ^{-1} \sum_{i=1} ^K \log |x-x_i| - \frac{1}{2} \log (1+|x| ^2) + K^{-1}\sum_{i=1} ^K \log |y-x_i| - \frac{1}{2} \log (1+|y|^2)\right) \\
 - \log\frac{|x-y|}{\sqrt{1+|x| ^2} \sqrt{1+|y| ^2}}. 
\end{multline}
The function in parenthesis on the second line is bounded below uniformly for $|x|,|y|\geq R$ large enough, so that, once integrated against 
$\one_{|x|\geq R} \one_{|y|\geq R}\sigma_n(x)\sigma_n(y)dxdy$ it gives a contribution uniformly bounded from below. For the contribution of the last term we note that 
\beq 
\frac{|x-y|}{\sqrt{1+|x| ^2} \sqrt{1+|y| ^2}} \leq \frac{|x|+|y|}{\sqrt{1+|x| ^2} \sqrt{1+|y| ^2}} \leq \frac{1}{\sqrt{1+|y| ^2}} + \frac{1}{\sqrt{1+|x| ^2}}
\eeq
and thus 
\begin{multline}
 \iint_{|x|\geq R, |y|\geq R} \log\frac{|x-y|}{\sqrt{1+|x| ^2} \sqrt{1+|y| ^2}} \sigma_n (x) \sigma_n (y)dxdy \\ 
 \leq -\frac{1}{2}\log (1+R ^2) \left(\int_{|x|\geq R} \sigma_n (x) dx\right)^2 + C. 
\end{multline}

All in all we deduce that, the energy of $\sigma_n$ being uniformly bounded above,  
\beq 
\frac{1}{4} \log (1+R ^2) \left(\int_{|x|\geq R} \sigma_n (x) dx\right)^2 \leq C,
\eeq
where $C$ depends neither on $n$ nor on $R>0$ provided the latter is large enough. Since $\log (1+R ^2) \to +\infty$ for $R\to \infty$, this implies tightness of the minimizing sequence, namely that for any $\eps >0$ there exists $R_\eps>0$ such that 
\beq \limsup_{n\to \infty} \int_{|x|\geq R_\eps} \sigma_n (x) dx < \eps.\eeq
\end{proof}
\medskip

To study the properties of the minimizer, whose existence and uniqueness we have now established,  we next consider the corresponding electrostatic potential.

\subsection{The Thomas-Fermi potential}\label{sec:TF pot}

To any electron density $\sigma$ we associate a potential 
\begin{equation}\label{eq:potential}
\Phi_\sigma = \Vnuc + \sigma * \log |\,.\,| = \left( \rhonuc - \sigma \right)*(- \log |\,.\,|).
\end{equation}
If $\int_{\R ^2} \sigma = K$ then this is the Coulomb potential generated by a neutral charge distribution. 
We prove first that it is regular away from the nuclei and decays at infinity. 

\begin{lemma}[\textbf{Regularity and asymptotic behavior of the potential}]\label{lem:decay pot}\mbox{}\\
If $\sigma \in \TFM$, the potential $\Phi_\sigma$  is continuous and  once continuously differentiable away from the nuclei and tends to zero at infinity. More generally, if $\sigma$ satisfies the conditions \eqref{eq:MTF} except neutrality, but $\int \sigma<K$, then 
\begin{equation}\label{eq:decay pot}
\left|  \Phi_\sigma (x) + \left(K - \int_{\R^2} \sigma \right) \log |x| \right| \to 0 
\end{equation}
 when $|x| \to \infty$
\end{lemma}

\begin{proof} 
If $\sigma \in \TFM$ has compact support  the regularity follows from ~\cite[Theorem~10.2]{LieLos-01}. The general case follow by cutting $\sigma$ at some finite radius $R$ and using the dominated convergence theorem as $R\to \infty$.

For the asymptotic behavior we  consider also first the case that $\sigma$ has support in some disk $D(0,R)$ with $R<\infty$.
For $|y|\leq \half |x|$ we have 
\beq\label{eq:split log} 
\log |x-y|=\log |x|+\log \left |1-\frac {y}{x} \right|\leq \log |x|+\frac{ |y/x|}{1-|y/x|}\leq \log |x| +2\frac{|y|}{|x|}
\eeq
where the ratio of two vectors in $\R^2$ is interpreted as that of the corresponding complex numbers.

Thus,  for $|x|\geq 2\max_i \{R, |x_i|\}$
\beq 
\left|  \Phi_\sigma (x) + \left(K - \int_{\R^2} \sigma \right) \log |x| \right|\leq 2\left( \sum_{i=1} ^K |x_i| + \int_{D(0,R)} |y| \sigma (y)  dy\right)\frac 1{|x|}.
\eeq
and the result is proved in this case.

Let us now only assume that 
\beq\int_{\R ^2} \sigma(y)\log (1+|y|)dy<\infty\eeq
as in the definition of $\TFM$. This implies that for any $\eps>0$ there is a $R_\eps$ such that 
\begin{equation}\label{eq:decay int}
\int_{|y|\geq R_\eps}\sigma(y)\log (1+|y|)dy<\eps. 
\end{equation}
For every $\eps>0$ there is a $\delta>0$ such that $\log(1-t)<\eps$ if $|t|<\delta$. Pick some large $R$ and write for $|x|\geq R \geq R_\eps$
\beq
\int_{\R^2} \log \left |1-\frac {y}x\right|\sigma(y)dy=\int_{|y|\leq \delta \cdot R} \log \left |1-\frac {y}x\right|\sigma(y)dy + \int_{|y|\geq \delta \cdot R}  \log \left |1-\frac {y}x\right|\sigma(y)dy.
\eeq
The first term is bounded by $({\rm const.})\,\eps$ by the choice of $\delta$ and because $\int_{\R ^2} \sigma$ is finite.

The second term we write as 
\beq
\int_{|y|\geq \delta \cdot R}\log|x-y|\sigma(y)dy -\log|x| \int_{|y|\geq \delta \cdot R}\sigma(y)dy.
\eeq
If $\delta \cdot R$ is large enough, we can ensure  
\beq\int_{|y|\geq \delta \cdot R}\sigma(y)dy\leq C \eps (\log|x|)^{-1}.\eeq 
because of~\eqref{eq:decay int}. In the remaining term (remember that $\delta \cdot R$ is large and $|x|\geq\delta \cdot R$) we split the integral into a disk $D(x,1)$ of radius 1 around $x$ and the rest. Since the logarithm is square integrable over the unit disk we can use the Cauchy-Schwarz inequality to conclude that 
\beq \left|\int_{D(x,1)}\log|x-y|\sigma(y)dy\right| \leq C \eps ^{1/2} \eeq
because $0\leq\sigma\leq 1$ and $\int_{D(x,1)}\sigma\leq \eps$. The integral over the complement of $D(x,1)$ is also small because there $\log|x-y|\leq 2\log|y|$ and 
the integral against $\sigma$ tends to zero as the size of the integration domain increases.

Summarizing, we have proved that for any $\eps >0$ there exists $R_\eps >0$ such that 
\beq
\int_{\R^2} \log \left| 1 - \frac{y}{x} \right| \sigma (y) dy \leq \eps
\eeq
for $|x| \geq R_\eps$. Returning to the first equality in~\eqref{eq:split log} proves the lemma.

\end{proof}

\begin{lemma}[\textbf{Variational inequalities}]\label{lem:var eq}\mbox{}\\
Recall the definition~\eqref{eq:TF pot} of the Thomas-Fermi potential. There exists $\lambda \in \R$ such that for almost every $x$ 
\begin{align}\label{eq:TF eq bis1}
\TFpot(x) &\geq \lambda \mbox{ if } \TFmin(x) =1  \\ \label{eq:TF eq bis2}
\TFpot(x)&= \lambda \mbox{ if } 0<\TFmin(x)<1\\ \label{eq:TF eq bis3}
\TFpot (x)&\leq \lambda \mbox{ if } \TFmin(x) = 0 
\end{align}
\end{lemma}
\begin{proof}
Eqs.\ \eqref{eq:TF eq bis1}--\eqref{eq:TF eq bis3} follow from local minimality of the functional \eqref{eq:TF func} at  $\TFmin$
by performing small variations in the standard way, c.f.~\cite[Lemma 4.1.2]{Friedman-88} or~\cite{BurChoTop-15,FraLie-16}. Here $\lambda$ is the Lagrange multiplier associated with the mass constraint. That we get only inequalities on the sets $\{ \TFmin =1  \}$ and $\{\TFmin = 0 \} $ comes from the fact that the variational constraint $0\leq \sigma \leq 1$ allows only one-sided variations on these sets.
\end{proof}

The Lagrange multiplier $\lambda$ we shall refer to as a chemical potential for it is related\footnote{We do not make this relation explicit for we don't need it here. See~\cite{LieSim-77b} where this is explained in the setting of usual Thomas-Fermi theory.} to variations of the minimal energy as a function of the electronic charge. We have the following: 

\begin{lemma}[{\bf Value of the chemical potential}]\label{lem:chem pot}\mbox{}\\
The Lagrange multiplier $\lambda$ defined in Lemma~\ref{lem:var eq} is equal to $0$ and $\TFpot = 0$ a.e.\ where $\TFmin=0$. 
\end{lemma}
\begin{proof}
Assume $\lambda<0$. By Lemma~\ref{lem:decay pot} $\TFpot$ vanishes at infinity so there is a radius $R$ such that $\TFpot(x)>\lambda$ for $|x|\geq R$. By \eqref{eq:TF eq bis1}--\eqref{eq:TF eq bis3} this implies $\TFmin(x)=1$ a.e. for $|x|\geq R$, contradicting the finiteness of $\int\sigma$. Hence $\lambda<0$ is excluded.

Assume  $\lambda>0$. As above, there is a radius $R$ such that $\TFpot(x)<\lambda$ for $|x|\geq R$. By \eqref{eq:TF eq bis1}--\eqref{eq:TF eq bis3} this implies $\TFmin(x)=0$ a.e. for $|x|\geq R$. Hence the set $\{\TFmin >0\}$ is bounded, up to a set of measure zero.
Consider the circular average 
\begin{equation}\label{eq:circav}
\Phiav (R) := \frac{1}{2\pi} \int_0^{2\pi}\TFpot\left(R\,e^{i\theta}\right)d\theta. 
\end{equation}
Because $\TFpot$ is harmonic on $\{x:\, \TFpot(x)<\lambda\}$, it cannot have a local minimum on this set so we know that $\TFpot(x)\geq 0$ everywhere and hence $\Phiav(R)\geq 0$. By continuity of $\TFpot$ and the assumption $\lambda>0$ there is a radius $R_0$ such that $\Phiav(R_0)>0$, but $\Phiav (R)<\lambda$ for $R>R_0$. The circle with radius $R_0$ thus encloses the whole set $\{\TFmin >0\}$, up to a set of measure zero. On the other hand, by Newton's Theorem~\cite[Theorem~9.7]{LieLos-01},  $\Phiav(R_0)$ is proportional to the total charge inside the circle of radius $R_0$ which is zero and we have a contradiction.

Having excluded both $\lambda <0$ and $\lambda >0$ we conclude that $\lambda=0$. Finally, $\{x:\, \TFpot(x)<0\}$ is empty because $\TFpot$ would be harmonic on this set and zero on its boundary.\end{proof}

With Lemmas \ref{lem:var eq}  and \ref{lem:chem pot} we have established  Item (2) of Theorem \ref{thm:TF theory} We next show that $\TFmin$ takes only the values $0$ or $1$ almost everywhere.

\begin{proof}[Proof of Theorem~\ref{thm:TF theory}, Item~(3)]
We want to show that the set  
\begin{equation}
A:=\left\{ x\,:\, 0< \TFmin(x) <1 \right\} 
\end{equation}
has zero Lebesgue measure.  By the definition~\eqref{eq:TF pot} of  $\TFpot$ we have in the sense of distributions
\begin{equation}\label{del1}
-\Delta \TFpot = 2\pi \left(\rhonuc - \TFmin\right).
\end{equation}
with $\rhonuc=\sum_{i=1}^K\delta_{x_i}$. 
Because $\TFpot$ is constant (in fact zero) on $A$, we can immediately exclude that $A$ contains an open set, because on such a set we would have
\begin{equation}\label{del2}
 \Delta \TFpot = 0
\end{equation}
contradicting \eqref{del1}.
Since there exist sets of positive measure but with empty interior this does not prove the lemma, however.
The problem of overcoming the lack of knowledge of $A$, and showing that 
\eqref{del2} holds a.e. on $A$ was solved in a similar problem in \cite{FraLie-16}
and we sketch that argument here. 

If a function $f$ is in $W^{1,1}_{\rm loc}$ then on the set on which $f$ has
any fixed value, its gradient, $\nabla f$, vanishes except for a set of measure zero, see~\cite{SerVar-69,MarMiz-72,AlmLie-89} or~\cite[Theorem~6.9]{LieLos-01}. Away from the nuclei, $\TFpot$ is explicitly seen to be continuous and in $L^1_{\rm loc}$, and the same holds for its gradient $\nabla \TFpot$.  On $A$ we thus have $\nabla \TFpot = 0$. In order to repeat this argument for the Laplacian, we need to show that $\TFpot$ is in $W^{2,1}_{\rm loc}$. 

What we do know is that the sum of the second
derivatives, $\Delta  \TFpot$, is an $L^1_{\rm loc}$ function (see
\eqref{del1}). What we need is that {\it all} mixed second derivatives are
in  $L^1_{\rm loc}$. This follows by an exercise involving Fourier transforms, c.f.\  \cite[Lemma 11]{FraLie-16}. Thus \eqref{del2}
does hold a.e. on $A$, and this contradicts \eqref{del1} unless $A$ has zero Lebesgue measure.
\end{proof}

We now turn to the last item of Theorem \ref{thm:TF theory}, concerning the support of 
the electron density.

\begin{proof}[Proof of Theorem~\ref{thm:TF theory}, Item~(4)]\mbox{}\\

\noindent\textbf{Step 1.} We first prove the estimate \eqref{eq:support TF} on the support of $\TFmin$. By Lemma \ref{lem:chem pot}
 we know that  $\lambda = 0$. There is an $R$ such that $|x_i|<R$ for all $i$. For $|x|\geq R$ and $\hat R>R$ we introduce the function 
\beq
g(x)=\begin{cases} \pi\cdot(\hat R-|x|)^2 &\hbox{ if } |x|\leq \hat R\\ 0&\hbox{ if } |x|>\hat R \end{cases}.
\eeq
Then, for $R<|x|\leq \hat R$,
\beq \Delta g\leq 2\pi.\eeq
Consider the set 
\beq\mathcal{S} = \{x:\,  0 \leq g(x) < \TFpot(x),\ |x| \geq R\}.\eeq
Since $\TFpot > 0$ on this set, which contains no nucleus by assumption, we have $\TFmin = 1$ and $\Delta \TFpot=2\pi$ on $\mathcal{S}$. Hence
\beq \Delta (g-\TFpot)\leq 0\eeq
on $\mathcal{S}$, so the function $g-\TFpot$ is superharmonic on $\mathcal{S}$ and reaches its minimum at the boundary or at infinity. At infinity the function vanishes. For $|x| = R$ we can make $g-\TFpot \geq 0$ by choosing 
\beq \hat{R} = R + \sqrt{M_R}\eeq
where $M_R$ is given by \eqref{eq:MR}. 
We conclude that on $\mathcal{S}$ we have $g-\TFpot \geq 0$ and thus this set must be empty. Then $\TFpot \leq g$ for $|x|\geq R$, so $\TFpot=0$ and hence $\sigma^{\rm TF}=(2\pi)^{-1}\Delta\TFpot=0$ for $|x|>\hat R$. This yields the desired result, Equation~\eqref{eq:support TF}.

\medskip

\noindent\textbf{Step 2.}
Next we consider the open set $\STF(x_1,\dots,x_K)\equiv \STF$ where the
potential $\TFpot$ is strictly positive. From Equation~\eqref{eq:TF eq} we
have $\STF\subset \{\sigma^{\rm TF}>0\}$ while $ \{\sigma^{\rm TF}>0\}$ is
contained in the closure $\bar\Sigma^{\rm TF}$. Since $\TFpot$ is
continuous outside the \lq\lq nuclei\rq\rq, $\STF$ is Jordan measurable which
means that its boundary has Lebesgue measure zero. Hence $\STF$ and $ \{\sigma^{\rm
TF}>0\}$  are equal up to a null set. Since $\sigma^{\rm TF}$ takes only
the values 1 or 0, it follows that the area of $\STF$ is equal to $K$.

\medskip

\noindent\textbf{Step 3.} Since $\TFpot (x) \to +\infty $ when $x$ approaches a location of a nucleus,  we have $x_i\in \STF$ for all $i$. Consider a connected component of $\STF$ which does not include any $x_i$. On this component we have 
\beq\Delta \TFpot = 2\pi\eeq
and thus $\TFpot$ is subharmonic, hence reaches its maximum at the boundary, where $\TFpot = 0$. This contradicts the strict positivity of $\TFpot$ on $\STF(x_1,\dots,x_K)$. Hence all connected components of $\STF$ contain at least one nucleus.

\medskip

\noindent\textbf{Step 4.} Denote by $\Phi_{K-1}^{\rm TF}$ and $\Phi_{K}^{\rm TF}$ the TF potentials corresponding to 
$\{x_1,\dots,x_{K-1}\}$ and $\{x_1,\dots,x_K\}$ respectively and by $\sigma_{K}$ and $\sigma_{K-1}$ the minimizing densities. Likewise we denote by $\STF_{K-1}$ and  $\STF_{K}$ the sets where the potentials $\Phi_{K-1}^{\rm TF}$ and $\Phi_{K}^{\rm TF}$ are strictly positive. Consider the open set
\beq
\mathcal U=\left\{x:\, \Phi_{K}^{\rm TF}(x)<\Phi_{K-1}^{\rm TF}\right\}.
\eeq
Clearly, $x_i\notin \mathcal U$ for all $i=1,\dots,K$. Since $\Phi_{K}^{\rm TF}\geq 0$ we have $\Phi_{K-1}>0$ on $\mathcal U$ and hence $\sigma_{K-1}=1$ on $\mathcal U$. Thus
\beq (2\pi)^{-1}\Delta(\Phi_{K-1}-\Phi_{K})=\sigma_{K-1}-\sigma_{K}\geq 0\eeq
on $\mathcal U$ because $\sigma_{K}\leq 1$. Thus $\Phi_{K-1}-\Phi_{K}$ is subharmonic on $\mathcal U$ and $\mathcal U$ must be empty. We conclude that
\beq \Phi_{K-1}(z)\leq\Phi_{K}(z)\eeq
everywhere. In  $\STF_{K-1}$ we have $\Phi_{K-1}(z)>0$, hence $\Phi_{K}(z)>0$, so 
$\STF_{K-1}\subset \STF_K$.

\medskip

\noindent\textbf{Step 5.} For a single nucleus, uniqueness implies that the minimizer must be rotationally symmetric around $x_1$. By Step 3 the support of the density is connected and hence a full disk with $\TFmin=1$.  The  Poisson equation then lead to the unique solution
\beq 
\TFpot(x)=\begin{cases}-\log|x-x_1| +\pi |x-x_1|^2-\half\log\pi-1&\hbox{ if } |x-x_1|< 1/\sqrt \pi\\  0 &\hbox{ if } |x-x_1|\geq 1/\sqrt \pi\\ \end{cases}.
\eeq
\end{proof}

\section{Local density upper bounds for classical ground states}\label{sec:class GS}

In this section we discuss the ground state of a generalized classical jellium Hamiltonian as presented in the introduction, and prove the announced local density upper bound on any microscopic scale $\gg 1$. The argument is split in two steps, corresponding to Subsections~\ref{sec:excl prop} and~\ref{sec:excl dens}:
\begin{itemize}
 \item First we prove an exclusion rule for ground states configurations: no point can lie inside a screening region defined by any subset of the other points. The latter is given by the open set $\STF$ where the TF potential $\Phi^{\rm TF}$ corresponding to the chosen subset of the ground state configuration is strictly positive. 
 \item Second, we prove that \emph{any} configuration of points satisfying the above exclusion rule must have a local density bounded above by $1$, which is the desired optimal upper bound.
\end{itemize}
In Subsection~\ref{sec:excl app} we apply these results to ground states of the Hamiltonian \eqref{eq:class Hamil 1} entering the plasma analogy, and small perturbations thereof.

\subsection{Exclusion by screening}\label{sec:excl prop}

The exclusion rule we shall chiefly rely on in the sequel is as follows: 

\begin{definition}[\textbf{Exclusion by screening}]\label{def:excl prop}\mbox{}\\
Let $X_N = \{x_1,\ldots,x_N\} \subset \R ^{2}$. Take any subset $\{y_1,\ldots,y_K,y_{K+1}\} \subset X_N$ of distinct points $y_i$. Define as in \eqref{eq:sigma TF}  $\STF(y_1,\ldots,y_K)$  to be the open set where the TF potential \eqref{eq:TF pot},
generated by nuclei at locations $y_1,\ldots,y_K$ and the corresponding background charge distribution $\sigma^{\rm TF}$, is strictly positive. We say that $X_N$ satisfies the {\it exclusion rule} if 
\begin{equation}\label{eq:excl prop}
y_{K+1} \notin \STF(y_1,\ldots,y_K)
\end{equation}
holds for any choice of $y_1,\ldots, y_K,y_{K+1}\in X_N$.\hfill\qed
\end{definition}

We call this ``exclusion by screening'' because  $\Sigma^{\rm TF}$ defines a region surrounding $y_1,\ldots,y_K$ outside of which the background charge density $\sigma^{\rm TF}$ totally screens the potential generated by $y_1,\ldots,y_K$. Indeed, as we proved in the preceding section, Equation~\eqref{eq:TF eq}, the total potential generated by the background and the point charges together vanishes identically outside $\STF$. Note also that the closure of $\STF$ is the essential support of $\sigma^{\rm TF}$.

The relevance of this notion to our problem is as follows:  Consider  a Hamiltonian function on $\R^{2N}$ of the form 
\begin{equation}\label{eq:Ham ref}
\cH(x_1,\dots, x_N)=\frac \pi 2\sum_{i=1}^N|x_i|^2-\sum_{1\leq i<j\leq N} \log{|x_i-x_j|}+\cW(x_1,\dots, x_N) 
\end{equation}
with $\cW$ symmetric and superharmonic in each variable $x_i$. Equation~\eqref{eq:Ham ref} is obtained from~\eqref{eq:class Hamil 1} by a simple change of length units that we perform in order that the maximum local density be $1$ instead of $\pi \ell$ (see Subsection~\ref{sec:excl app} below).

\begin{proposition}[\textbf{Classical ground states satisfy the exclusion rule}]\label{pro:GS exclu}\mbox{}\\ 
Any minimizing configuration $\{x_1^0,\dots,x_N^0\}$ for $\cH$ satisfies the exclusion rule \eqref {eq:excl prop}.
\end{proposition}

\begin{proof}
By symmetry of the Hamiltonian we may, without loss, choose 
\beq y_i = x_i ^0, 1\leq i \leq K+1.\eeq
Consider fixing all points but $x_{K+1}^0$. The energy to consider is then
\begin{equation} 
G(x)= \cH(x_1^0,\dots,x_{K}^0,x,x_{K+2}^0,\cdots x_N^0).
\end{equation}
We claim that if $x\in \STF (x_1^0,\dots,x_{K}^0) \equiv \STF$ then there is an $\tilde x\in \partial  \STF$ such that $G(\tilde x)<G(x)$. Thus the minimizing point $x_{K+1} ^0$ cannot lie in $\STF$.

To prove the claim, we add and subtract a term $-\int_{\STF}\log{|x-x'|}dx'$ to write
\begin{align} \label{potential}
G(x)&=\Phi(x)+R(x)\\
\Phi(x)& =-\sum_{i=1}^K\log{|x-x_i^0|}+\int_{\STF} \log{|x-x'|}dx'\\
R(x)&=\frac{\pi}{2} |x|^2-\int_{\STF} \log{|x-x'|}dx' -\sum_{i=K+2}^N \log{|x-x_i^0|} + \cW(x) + \hbox{const.}
\end{align}
with a superharmonic function
\beq \cW (x) := \cW(x_1^0,\dots,x_{K}^0,x,x_{K+2}^0,\cdots x_N^0).\eeq
Now, $\Phi$ is precisely the TF potential corresponding to \lq\lq nuclear charges\rq\rq\ at $x_i^0,\dots x_K^0$. Hence, using~\eqref{eq:TF eq}, $\Phi>0$ on $\STF$ and zero on the boundary $\partial \STF$. 
The first two terms in $R$ are harmonic on $\STF$ when taken together. (The Laplacian applied to the first term gives $2\pi$ and to the second term $-2\pi$ on $\STF$.) The other terms are superharmonic on~$\STF$. Thus, $R$ takes its minimum on the boundary, so there is a $\tilde x\in  \partial \STF$ with $R(x)\geq R(\tilde x)$ for all $x\in \STF$. On the other hand, $\Phi(x)>0=\Phi(\tilde x)$ so $G(x)>G(\tilde x)$.
\end{proof}

\subsection{Exclusion rule and density bounds}\label{sec:excl dens}

We now show that any configuration of points satisfying the exclusion rule~\eqref{def:excl prop} has its local density everywhere bounded above by $1$. To this end let $R>0$ and define $n(R)$ to be the maximum number of nuclei that a disk $D(a,R)$ of radius $R$ can accommodate while respecting the exclusion rule:
\begin{equation}\label{eq:def max number}
n (R) =  \sup \left\{ N \in \N \big|  \mbox{\ there exists } X_N=\{x_1,\dots,x_N\} \subset D(a,R) \mbox{ satisfying}~\eqref{def:excl prop}  \right\}.  
\end{equation}
It is clear that $n (R)$ is independent of $a$ because a translation of the $x_i$ just translates the corresponding exclusion region.

Our density bound is as follows:

\begin{theorem}[\textbf{Exclusion rule implies a density bound}]\label{thm:mainconj}\mbox{}\\
We have
\beq\label{eq:mainbound}
\rhoM:= \limsup_{R\to \infty}\frac {n(R)}{\pi R^2}\leq 1.
\eeq
\end{theorem}
\medskip

We first recall from~\cite{RouYng-15}  a simpler density bound based on an unpublished  theorem of Lieb which is essentially \eqref{eq:excl prop} for the special case $K=1$.

\begin{lemma}[\textbf{Minimal distance between points}]\label{lem:min dist}\mbox{}\\
Let $\{x_1,\ldots,x_N\}$ be a configuration of points satisfying ~\eqref{def:excl prop}. Then 
\begin{equation} \label{eq:min dist}
\min_{1\leq i \neq j \leq N} |x_i-x_j| \geq \frac{1}{\sqrt{\pi}}.
\end{equation}
Consequently
\beq\label{eq:mainbound pre}
\rhoM = \limsup_{R\to \infty}\frac {n(R)}{\pi R^2}\leq 4.
\eeq
\end{lemma}

\begin{proof}
By Theorem \ref{thm:TF theory}, Item 4(v),  the exclusion region $\STF(x)$ for a single point $x$ is the open disk of radius $1/\sqrt \pi$ centered at $x$. Hence~\eqref{eq:min dist} follows from \eqref{eq:excl prop}. 

To deduce~\eqref{eq:mainbound pre}, pick a configuration of points $\{x_1,\ldots,x_N\}$ contained in a disk of radius $R$ and smear each point over an open disk of radius $1/(2\sqrt{\pi})$. Because of~\eqref{eq:min dist} these disks do not overlap and must all be contained in a slightly larger disk of radius $R+ O(1)$ for large $R$. Thus 
\beq \label{eq:main bound pre pre}
\frac{N}{4} \leq \pi R ^2 + O(R) 
\eeq
where the left side is the area covered by the small disks and the right-hand side the area of the large disk. This gives the desired result.
\end{proof}

The rough bound Equation~\eqref{eq:mainbound pre} turns out to be quite useful for the proof of the full Theorem~\ref{thm:mainconj}.  The first step is a lemma saying that the $\limsup$ in~\eqref{eq:mainbound} is attained by sequences of configurations with no large vacancies:

\begin{lemma}[\textbf{Maximal configurations have no vacancies}]\label{lem:nice sequence}\mbox{}\\
Let $R_k$ be a sequence of radii with $R_k\to\infty$ for $k\to\infty$ and $X_{N_k} = \{x_1,\ldots,x_{N_k}\} \subset D(0,R_k)$ a sequence of configurations  such that 
\begin{equation}\label{eq:nice sequence 1}
\frac{N_k}{\pi R_k ^2} \underset{k\to \infty}{\to} \rhoM.
\end{equation}
Then, for any fixed $\eps >0$ and any sub-disk $D(a,\eps R_k) \subset D(0,R_k)$,
\begin{equation}\label{eq:nice sequence 2}
\sharp \left\{ x \in X_{N_k} \cap D(a,\eps R_k) \right\} \geq \rhoM \pi\, \eps^2  R_k ^2 \left( 1+ o(1) \right)
\end{equation}
with $o(1)\to 0$ for $R_k\to\infty$.
\end{lemma}

This means that, for a configuration achieving the tightest packing in a given sequence of balls, the density is asymptotically  uniform on any length scale comparable to the size of the balls. 

\begin{proof} The main idea is that a density lower than $\rhoM$ in some sub-disk would have to be compensated for by a density higher than $\rhoM$ in another sub-disk. This would contradict the definition of $\rhoM$ as the $\limsup$ of the maximal density achievable in any sequence of disks with radii tending to $\infty$.

Note first that by  Lemma~\ref{lem:min dist} we have for any $R$, as in~\eqref{eq:main bound pre pre},    
\beq n (R) \leq \max \left( 1 , 4 \pi R ^2 + C R \right).\eeq
Hence $\rhoM$ is finite and, in fact, 
\beq \rhoM \leq  4.\eeq

We now tile the disk of radius $R_k$ with disjoint smaller disks of radius $s:=\eps R_k < 1$, labeled by an index $i=1 \ldots I$. We can achieve this leaving only an area $O(\eps ^2 R_k ^2)$ untiled. The latter we tile with disjoint still smaller disks of radius $t:=\eps ^{1/2} R_k ^{1/2}$, labeled by $j=1 \ldots J$. In this way we leave only a domain $A_{\rm h}$ with area $|A_{\rm h}|=O(\eps R_k)$ uncovered. Let $n\left(i,s\right)$, respectively $n\left(j,t\right)$, be the number of points of the configuration $X_{N_k}$ in each of the the smaller and the still smaller disks and let $n_{\rm h}$ be the number of points in $A_{\rm h}$. By definition of $\rhoM$ we have 
\begin{align}
n(R_k) &= \sum_{i=1} ^I n\left(i,s\right) + \sum_{j=1} ^J n\left(j,t\right) + n_{\rm h} \nonumber \\
&\leq n\left(1,s\right) + \sum_{i=2} ^I \rhoM \pi \eps ^2 R_k ^2 (1+o(1)) + \sum_{j=1}^J \rhoM \pi \eps R_k (1+o(1)) + C \max(1,|A_{\rm h}|)\nonumber \\
&= n(1,s) + \rhoM (1+o(1)) \left( \pi R_k ^2 - \pi \eps ^2 R_k ^2 - |A_{\rm h}| \right) +  C \max(1,|A_{\rm h}|).
\end{align}
In the second line we used  the definition of $\rhoM$ as a $\lim\sup$ and that for each fixed $\eps>0$, $\eps R_k\to\infty$. Moreover, we used ~Lemma~\ref{lem:min dist} to bound the density in $A_{\rm h}$. In the third line we simply used that we tile with disjoint sets. Since by assumption
\beq  n(R_k) =\rhoM \pi R_k ^2  (1+o(1)) \eeq
we deduce 
\beq n(1,s) \geq \rhoM\pi \eps ^2 R_k ^2 (1+o (1)) \left( 1 - C \eps ^{-1} R_k ^{-1} \right) + o(1)R_k  ^2. \eeq
For  large $R_k$ and fixed $\eps>0$, the small disk labeled $i=1$ must thus satisfy~\eqref{eq:nice sequence 2}. For any disk as in the statement, we can construct a tiling such that it is the first small disk and the result follows.
\end{proof}

 We now describe briefly the \emph{main arguments} leading to the proof of Theorem~\ref{thm:mainconj}.
The area of the ``neutralizing region'' $\STF$ associated with the nuclear charges contained in a disk of radius $R$ is equal to the number of charges in the disk, so  the task is to show that this area is asymptotically equal to the area of the disk as $R\to\infty$. On the boundary of the  region $\STF$ the potential generated by the point charges in the disk and the TF neutralizing density $\TFmin$ must vanish because of the exclusion rule. In particular,  the potential must vanish at any point of the nuclear charge configuration outside the disk. By the previous lemma, we can assume that the density of such vanishing points is bounded below uniformly, and combining this with a bound on the  gradient of the potential  we obtain an upper bound on the potential outside the disk.
We can now proceed  in two different ways (the second one was sketched in~\cite{LieRouYng-16}): 

\begin{itemize}
  \item The bound on the potential outside the disk together with~\eqref{eq:support TF} implies that we can enclose $\STF$ completely in a slightly larger disk whose radius behaves as $R (1+o(1))$. Thus the area of $\STF$, which is equal to the number of nuclei in the disk increases at most like the area of the disk plus a small correction. 
  \item Using Newton's theorem we obtain a bound on the circular average of the potential which involves the deviation from neutrality for the charge distribution within the disk. A comparison with the previous upper bound again leads to the desired result: Asymptotically the positive charge cannot be larger than the negative charge in the disk.
  \end{itemize}
  
 The details are as follows, the first steps being  common to both arguments.

\begin{proof}[Proof of Theorem~\ref{thm:mainconj}]
Consider any maximal density sequence of radii and configurations defined in the previous Lemma \ref{lem:nice sequence}. 
We use the lemma in two ways. First, taking $\eps=1/2$, we conclude that the sequence of radii 
$r_k:=R_k /2$ and the corresponding number $n_k$ of points $\{x_1,\dots,x_{n_k}\}$ in  $X_{N_k}\cap D(0,r_k)$  satisfy
\beq\label{eq:dens lim} 
\lim_{k\to\infty}\frac {n_k}{\pi r_k^2}=\rhoM.
\eeq
Second, for every $\eps>0$, the lemma implies that every point in the annulus 
\begin{equation}\label{eq:working annulus}
\mathcal A:=D(0,R_k)\setminus D(0,r_k) 
\end{equation}
is at most a distance $O(\eps r_k)$ from some point in $X_{N_k}\cap \mathcal A$ as $r_k\to\infty$.  

Let $\sigma^{\rm TF}_k$ be the TF density  defined by the point charges $\{x_1,\dots,x_{n_k}\}\subset D(0,r_k)$ and 
\beq 
\Phi^{\rm TF}_k = \left(\sum_{j = 1} ^{n_k} \delta_{x_j}- \TFmin_k \right)*(-\log |\,.\,|)
\eeq
the corresponding TF potential. From Prop.\ \ref{pro:GS exclu}
we know that
$\Phi^{\rm TF}_k$ vanishes at all points in $X_{N_k}\cap\mathcal A$. Moreover, we know that 
\beq 
\int_{\Sigma^{TF}_k} \sigma^{\rm TF}_k=n_k.
\eeq

Pick some $\half >\delta>0$. 
For $|x|\geq (1+\delta)r_k$ a simple estimate gives a bound on the gradient of the potential:
\begin{equation}\label{eq:bound grad pot}
|\nabla \Phi^{\rm TF}_k| \leq C \frac{n_k}{\delta\cdot r_k}. 
\end{equation}
Indeed, if $|x| \geq (1+\delta) r_k $ then $|x-x_j| \geq \delta\cdot r_k$ for $j=1,\dots,n_k$ and thus
\beq\sum_{j=1} ^{n_k} \frac{1}{|x-x_j|}  \leq \frac{n_k}{\delta\cdot r_k}. \eeq
Also, splitting the integral into two regions where $|x-y|\geq \alpha = \sqrt{n_k}$ or vice-versa,  
\begin{align}
\int_{y\in \R^2} \frac{1}{|x-y|} \TFmin_k (y) dy &\leq  \int_{|x-y|\leq \alpha} \frac{1}{|x-y|} dy + \alpha ^{-1} \int_{\R ^2} \TFmin_k \nonumber \\
&\leq C \alpha + n_k  \alpha^{-1} \leq C \sqrt{n_k} \leq C\frac{n_k}{\delta\cdot r_k}
\end{align}
because $\int \TFmin_k = n_k$ and we know from~\eqref{eq:dens lim} that $n_k \geq C r_k ^2$.

Since every point in the annulus $\mathcal A$ is at most a distance $O(\eps r_k)$ from a point where $\Phi^{\rm TF}_k$ vanishes, the gradient estimate \eqref{eq:bound grad pot} implies a bound on the potential in the annulus:
\begin{equation}\label{eq:up bound pot}
\sup_{x\in\mathcal A} |\Phi_k^{\rm TF} (x)| \leq C \eps r_k \frac{n _k }{\delta\cdot r_k} \leq C \frac{\eps}{\delta} r_k ^2 (1+o(1))
\end{equation}
where in the second inequality we used that $n_k$ is in any case smaller than $4 r_k ^2 (1+o(1))$ for large $k$ by Lemma~\ref{lem:min dist}. With~\eqref{eq:up bound pot} at our disposal, we can conclude in two different ways:

\medskip

\noindent\textbf{First proof.} By Equation~\eqref{eq:support TF} the bound \eqref{eq:up bound pot} implies a bound on the largest radius $\bar r_k$ of a disk $D(0,\bar r_k)$ containing the support of $\sigma_k^{\rm TF}$:
\beq\rb_k \leq \left((1+\delta) r_k + C  r_k \sqrt{\frac{\eps}{\delta}}\right) (1+o(1)).\eeq
Choosing $\delta = \eps ^{1/3} $ to optimize the above we obtain 
\begin{equation}\label{eq:enclose up}
 \rb_k \leq r_k (1+ C \eps ^{1/3}) (1+o(1)). 
\end{equation}
Since $0\leq \sigma^{\rm TF}_k\leq 1$ this implies 
\beq n_k=\int_{\R ^2} \sigma^{\rm TF}_k\leq \pi\bar r_k^2\leq \pi r_k^2(1+ C \eps ^{1/3})^2 (1+o(1))\eeq
and thus, since $\eps>0$ is arbitrary,
\beq \rhoM=\lim_{k\to\infty}\frac {n_k}{\pi r_k^2}\leq 1.\eeq

\medskip

\noindent\textbf{Second proof.} Let $\Phiav$ be the circular average \eqref{eq:circav} of $\TFpot_k$. By Newton's theorem~\cite[Theorem~9.7]{LieLos-01} we have in the complement of $D(0,r_k)$ 
$$ \partial_r \Phiav = \frac{ M(r) - n_k}{r} \: \mbox{ with } \: M(r) = \int_{D(0,r)} \TFmin.$$
Thus, picking some fixed $\delta'>\delta>0$,
$$ \Phiav ((1+\delta')r_k) - \Phiav ((1+\delta)r_k) \leq  \left(\pi r_k ^2 (1+\delta') ^2 - n_k \right) \log \frac{1+\delta'}{1+\delta}$$
because $M(r) \leq \pi r^2$ is an increasing function. By construction we have, for large $k$, $n_k \geq \pi r_k ^2 \rhoM (1+o(1))$ and so
$$ \Phiav ((1+\delta')r_k) - \Phiav ((1+\delta)r_k) \leq  \pi r_k ^2 \left(  (1+\delta') ^2 - \rhoM - o(1) \right) \log \frac{1+\delta'}{1+\delta}.$$
Combining with~\eqref{eq:up bound pot} we deduce 
$$ \pi r_k ^2 \left(  (1+\delta') ^2 - \rhoM - o(1) \right) \log \frac{1+\delta'}{1+\delta} \geq - C \frac{\eps}{\delta} r_k ^2 $$
and thus 
$$ \rhoM \leq (1+\delta') ^2 (1+o(1)) + C \left( \log \frac{1+\delta'}{1+\delta} \right) ^{-1} \frac{\eps}{\delta} $$
for any $\eps >0$ and $\delta' > \delta > 0$ where the $o(1)$ goes to when $k\to \infty$. We may thus take the limits $k\to \infty, \eps \to 0, \delta \to 0, \delta' \to 0$, in this order to deduce that $\rhoM \leq 1$ as desired. 

\end{proof}

\subsection{Applications}\label{sec:excl app}

After scaling, $x\to z =  \sqrt{\pi\ell}\,x$, Theorem \ref{thm:mainconj} applies to the Hamiltonian
\begin{equation}\label{plasmaham}
H_N (Z_N) = \sum_{j=1} ^N |z_j| ^2 - 2\ell \sum_{1 \leq i<j \leq N } \log {|z_i-z_j|} + W(Z_N)
\end{equation}
for any superharmonic function $W$ of the $N$ variables. To obtain a corresponding result for low-temperature Gibbs states, we shall later use a Feynman-Hellmann argument and thus need to obtain bounds for a perturbed version of~\eqref{plasmaham}:
\begin{equation} \label{eq:pert ham}
H_N^\eps(Z_N)=H_N(Z_N)+\eps\sum_{i=1}^N U(z_i)
\end{equation}
where $\eps >0$ is a small enough number.

\begin{proposition}[\textbf{Density bound for perturbed plasma ground states}]\label{pro:dens plasma}\mbox{}\\
Assume $U\in C^2 (\mathbb R^2)$ with $\Delta U$ uniformly bounded on $\R^2$. Let $Z_N ^0 = (z_1^0,\dots,z_N^0)$ be a minimizing configuration for~\eqref{eq:pert ham} and 
\beq 
\mu^0_\eps(z)=\frac 1N\sum_{i=1}^N\delta(z-z_i^0)
\eeq 
the corresponding empirical measure. For any open set with Lipschitz boundary $\Omega$, denote by $\Omega_r$ the set obtained by dilating $\Omega$ around some origin by a factor $r \underset{N \to \infty}{\longrightarrow} \infty$. Then, as $N\to \infty$
\beq \label{eq:pert dens bound}
\int_{\Omega_r} \mu^0_\eps \leq \frac 1{N\ell \pi}|\Omega_r|\left( 1 + \frac {\eps}4 \Vert \Delta U\Vert_{\infty} \right)(1+o(1))
\eeq
with 
$ \Vert \Delta U \Vert_\infty = \sup_{\R^2} |\Delta U|$.
\end{proposition}

\begin{proof}
For $\eps = 0$ and $\Omega$ a disk this is just a combination of Proposition~\ref{pro:GS exclu} and Theorem~\ref{thm:mainconj}, together with the scaling $x\to \sqrt{\pi\ell}\,z$. For the perturbed Hamiltonian, we add and subtract 
\beq
\frac \eps 4\sum_{i=1}^N\Vert\Delta U\Vert_\infty|z_i|^2
\eeq
and use the fact that $U(z)-\hbox{$\frac 14$} \Vert \Delta U\Vert_\infty|z|^2$ is superharmonic. Hence
\beq\eps \sum_{i=1}^N\left(U(z_i)-\frac 14 \Vert\Delta U\Vert_\infty|z_i|^2\right)\eeq
can be absorbed in $W$ and we only need to reproduce the proof by writing~\eqref{eq:pert ham} in the manner  
\begin{equation}
H_N^\eps (Z_N) = \left( 1 + \frac{\eps}{4} \Vert\Delta U\Vert_\infty \right) \left[\sum_{j=1} ^N  |z_j| ^2 - {2\tilde \ell} \sum_{1 \leq i<j \leq N }
\log{|z_i-z_j|} + \widetilde{W}(Z_N)\right].
\end{equation}
where $\tilde \ell= \left( 1 + \frac{\eps}{4} \Vert\Delta U\Vert_\infty \right)^{-1}\ell$ and $\widetilde{W}$ is again superharmonic in each variable.

To obtain the result for a general domain $\Omega_r$, we argue by covering it with disks, as in the \lq\lq Cheese Theorem\rq\rq~\cite[Section~14.4]{LieSei-09} which quantifies how effectively a region can be approximately covered with non-overlapping disks. We tile $\Omega_r$ with disks of radius $\sqrt{r}$, leaving an open set of area at most $O(r)$ untiled. In the disks, we argue as previously. In the remaining untiled area we can use the rougher bound on the minimal separation of points, Lemma~\ref{lem:min dist} to show that it contains at most $O (r)$ points. This is of smaller order than the main contribution to~\eqref{eq:pert dens bound}, which comes from the disks, whose union contains $O(r^2)$ points, and the result follows.
\end{proof}

\section{Local density bounds for classical Gibbs states}\label{sec:dens Gibbs}

Here we prove our main result Theorem~\ref{thm:main}. We proceed in two steps which are roughly as follows:
\begin{itemize}
 \item we first obtain the result for fully-correlated states having angular momentum bounded by $CN^2$ for some $C>0$. These have the bulk of their density contained in a region that is not vastly larger than the extension of the Laughlin state.
 \item next we eliminate the a priori angular momentum bound by a localization procedure.
\end{itemize}

\subsection{Conditional local density bound}\label{sec:cond bound}

In order to combine neatly with the localization procedure used in Subsection~\ref{sec:localization} below, we must work in a slightly more general setting than discussed in Section~\ref{sec:proof out}. Let $\mubf$ be a $N$-body probability density that we can write in the form 
\begin{equation}\label{eq:mubf}
\mubf (z_1,\ldots,z_N)= \left| \PsiLau (z_1,\ldots,z_N) \right| ^2 \prod_{j=1} ^N\exp\left( - \frac{|z_j-a|^2}{L^2} \right) G(z_1,\ldots,z_N)  
\end{equation}
where $L>0$ is a number and $\log G$ is subharmonic (see e.g.~\cite[Chapter~9]{LieLos-01} for equivalent characterizations): 
\begin{equation}\label{eq:PL class} 
G \geq 0 \mbox{ and } \Delta_j \log G \geq 0 \mbox{ for any } j=1 \ldots N. 
\end{equation}
Note that $G$ need not be the modulus of an analytic function here. It could, in fact, be of the form $G =  |F_1| + |F_2| $ with $F_1$ and $F_2$ analytic, and we shall use this later, see Lemma~\ref{lem:PL class} below. 

Marginal probabilities of $\mubf$ will be denoted as 
\begin{equation}\label{eq:marginals}
\mubf ^{(k)} (z_1,\ldots,z_k)= \int_{\R^{2(N-k)}} \mubf (z_1,\ldots,z_N) d z_{k+1}\ldots dz_{N}.
\end{equation}
In this subsection we prove the following on the first marginal density:

\begin{theorem}[\textbf{Conditional density bound}]\label{thm:densitybound}\mbox{}\\
Take $L\geq 1$ in~\eqref{eq:mubf}. For any disk $D$ of radius $r= N ^\alpha$, $\alpha > 1/4$, any choice of origin $a\in \R^2$
\begin{equation}\label{avdensitybound bis}
\int_D\mubf^{(1)} \leq \frac 1{N\ell\pi}|D| \left(1+o_D (1) + O (L ^{-2} ) \right) + C N ^{-1 - 2 \alpha'} \intR |z-a| ^2 \mubf^{(1)} (z)dz 
\end{equation}
for any $\alpha'< \alpha$ where $|D|$ is the area of the disk and $o_D(1)$ tends to zero as $|D|\to\infty$.
\end{theorem}

To see how this relates to our main result, think of the case where $\mubf =|\Psi_F| ^2 $ is the density of a fully-correlated state, that is where $L = +\infty$ and $G = |F|$ with $F$ analytic. Assuming a total angular momentum of order $N^2$
\begin{equation}\label{eq:ang mom bound}\left\langle \Psi_F\big| \sum_{j=1} ^N z_j \dd_{z_j} - \bar{z}_j \dd_{\bar{z}_j} \big| \Psi_F \right\rangle \leq C N^2 
\end{equation}
we have (see e.g.~\cite[Equation~(2.18)]{RouSerYng-13b})
\begin{equation}\label{eq:a priori mom}
 \intR |z| ^2 \mubf^{(1)} (z)dz \leq C N. 
\end{equation}
Theorem~\ref{thm:main} follows in this case by taking $a=0$: the main (first) term behaves as $N^{- 1 + 2\alpha} \gg N ^{-1 + 2\alpha'}$ and the error (second term) as $N^{-2 \alpha'}$, thus the statement is meaningful for $\alpha > 1/4$.

\begin{proof}
The method is basically the same as in~\cite[Section~3.1]{RouYng-15}. We first write $\mubf$ as a Boltzmann-Gibbs factor:
\begin{align}\label{eq:hamil proof}
\mubf (z_1,\ldots,z_N)&= \frac{1}{\cZ_N}\exp \left( - \bH_N (z_1,\ldots,z_N) \right)\nonumber\\
\bH_N (Z_N) &= \sum_{j=1} ^N \left( |z_j| ^2 +\frac{|z_j-a| ^2}{L^2} \right) - 2\ell \sum_{1\leq i < j \leq N} \log |z_i-z_j| - \log G (z_1,\ldots,z_N) 
\end{align}
with $\cZ_N$ a normalization constant. Thus $\mubf$ minimizes the free-energy functional
\begin{equation}\label{eq:free ener proof 1}
\cE_N ^0 [\nubf] := \int_{\R ^{2N}} \bH_N (Z_N) d\nubf(Z_N) + \int_{\R^{2N}} \nubf \log \nubf
\end{equation}
over all $N$-particles probability densities $\nubf$. We denote by $E_N ^0 = - \log \cZ_N$ the minimum value.

For $r,\delta>0$ let $U_r$ denote the negative characteristic function of a disk with radius $r$. Let $U_{r,\delta}$ be a regularization over a distance $\delta$ so that 
\begin{equation}\label{eq:Deltabound}
\Vert\Delta U_{r,\delta}\Vert_\infty\leq C \delta^{-2}.
\end{equation}
We may choose the regularization in such a way that
\begin{equation} \label{eq:Deltabound 2}
U_{r+\delta}\leq U_{r,\delta}\leq U_r.
\end{equation}
Consider now a perturbed version of the above classical Hamiltonian $\bH_N$: 
\begin{equation}\label{eq:pert proof}
\bH_N ^\eps (Z_N) = \bH_N (Z_N) + \eps \sum_{j=1} ^N U_{r,\delta} (z_j)
\end{equation}
along with the associated free-energy functional 
\begin{equation}\label{eq:free ener proof 2}
\cE_N ^\eps [\nubf] := \int_{\R ^{2N}} \bH_N ^\eps (Z_N) d\nubf(Z_N) + \int_{\R^{2N}} \nubf \log \nubf
\end{equation}
whose minimum we denote by $E_N ^\eps $. By definition we have 
\beq
\cE_N ^{\eps} [\mubf] = E_N^0 + \eps N \int_{\R ^2 }U_{r,\delta} \mubf ^{(1)}
\eeq
and we will obtain our density estimate from upper and lower bounds to this free energy.

\medskip

\noindent\textbf{Free-energy upper bound.} We need an upper bound on $E_N ^0$. We use the trial state 
\begin{equation}\label{eq:trial state}
\mubf ^t (z_1,\ldots,z_N) := \left(\frac{1}{\pi \eta ^2}\right) ^N \one_{z_1 \in D(z_1 ^0,\eta)} \otimes \ldots \otimes \one_{z_N \in D(z_N ^0,\eta)}. 
\end{equation}
where  $Z_N ^0 = (z_1 ^0,\ldots,z_N ^0)$ is a minimizing configuration for $\bH_N$. We have 
\beq \intRN \mubf ^t \log \mubf ^t = - N \log (\pi \eta ^2)\eeq
and 
\begin{multline}
 \int_{\R ^{2N}} \bH_N (Z_N) \mubf ^t (Z_N) dZ_N \leq \sum_{j=1 }^N \left( \frac{1}{\pi \eta ^2} \int_{D(z_j ^0,\eta)} |z| ^2 dz + \frac{1}{L ^2 \pi \eta ^2} \int_{D(z_j ^0,\eta)} |z_j - a| ^2 dz \right) \\ 
 - 2\ell \sum_{1 \leq i<j \leq N } \log | z_i ^0-z_j ^0 | - \log G (Z_N ^0), 
\end{multline}
because superharmonicity in each variable of the function 
\beq(z_1,\ldots,z_N) \mapsto - 2\ell \sum_{1 \leq i<j \leq N } \log | z_i -z_j  | - \log G (Z_N )\eeq
implies that it cannot increase upon taking an average over disks centered at the $z_i ^0$'s (see e.g.~\cite[Chapter~9]{LieLos-01}). For the one-body term we have, integrating in polar coordinates 
\beq\frac{1}{\pi \eta ^2} \int_{D(z_j ^0,\eta)} |z| ^2 dz = \frac{1}{\pi \eta ^2} \int_{D(0,\eta)} |z+z_j ^0| ^2 dz = |z_j ^0| ^2 + \frac{1}{\pi \eta ^2} \int_{D(0,\eta)} |z| ^2 dz = |z_j ^0| ^2 + \frac{\eta ^2}{2}\eeq
and by the same token, since $L\geq 1$
\beq \frac{1}{L ^2 \pi \eta ^2} \int_{D(z_j ^0,\eta)} |z_j - a| ^2 dz \leq \frac{1}{L ^2}|z_j ^0-a| ^2 + \frac{\eta ^2}{2}.\eeq
Fixing $\eta$ we thus deduce 
\begin{equation}\label{eq:free ener up bound}
\cE_N ^\eps [\mubf ] \leq \min_{\R ^{2N}} \bH_N^0 + \eps N \intR U_{r,\delta} \mubf ^{(1)} + C N.
\end{equation}

\medskip

\noindent\textbf{Free-energy lower bound.} By positivity of the relative entropy (Jensen's inequality) we have, for any pair of probability measures $\nubf,\nubf ^t$, 
\begin{equation}\label{eq:rel entr}
\intRN \nubf \log \nubf \geq \intRN \nubf \log \nubf ^t. 
\end{equation}
Taking
\beq  \nubf ^t (z_1,\ldots,z_N) := (\pi N)^{-N} \prod_{j=1} ^N \exp\left(- \frac{|z_j - a| ^2}{N} \right) \eeq
we deduce that for any probability measure 
\beq \intRN \nubf \log \nubf  \geq - \intR |z-a| ^2 \nubf ^{(1)} (z) dz - O (N \log N).\eeq

Thus 
\begin{align}
\cE_N ^{\eps} [\mubf] &\geq \intRN \bH_N^\eps (Z_N) \mubf (Z_N) dZ_N - \intR |z-a| ^2 \mubf ^{(1)} (z) dz - O (N \log N) 
\\
&\geq \min_{\R ^{2N}} \bH_N^\eps - \intR |z-a| ^2 \mubf ^{(1)} (z) dz - O (N \log N)\\
&\geq \min_{\R ^{2N}} \bH_N^0 + \eps N \intR U_{r,\delta} \mu^0_\eps - \intR |z-a| ^2 \mubf ^{(1)} (z) dz - O (N \log N)
\end{align}
where $\mu^0_\eps$ is the empirical measure of a minimizing configuration for $\bH_N ^{\eps}$. 

\medskip

\noindent\textbf{Conclusion.} Combining with~\eqref{eq:free ener up bound} we deduce 
\beq \intR U_{r,\delta} \mubf ^{(1)} \geq \intR U_{r,\delta} \mu^0_\eps - \left(\eps N \right) ^{-1} \intR |z-a| ^2 \mubf ^{(1)} (z) dz -  C \eps ^{-1} \log N. \eeq
But, by~\eqref{eq:Deltabound 2} we have 
\beq \int_{D_r} \mubf ^{(1)} = - \intR U_r \mubf^{(1)} \leq - \intR U_{r,\delta} \mubf^{(1)}\eeq
and 
\beq -\intR U_{r,\delta} \mu^0_\eps \leq \int_{D_{r+\delta}} \mu_\eps ^0.\eeq
Proposition~\ref{pro:dens plasma} applies to $\mu_0 ^\eps$ because it is the empirical measure of a minimizing configuration for~\eqref{eq:hamil proof} (compare with~\eqref{eq:pert ham}). All in all we thus have 
\beq \int_{D_r} \mubf ^{(1)} \leq \pi N ^{-1} (r + \delta) ^2 \left( 1 + o (1) + \frac{C}{L^2} + \eps \delta ^{-2}\right) + \eps ^{-1} \log N +  \left(\eps N \right) ^{-1} \intR |z-a| ^2 \mubf ^{(1)} (z) dz \eeq
where we have used~\eqref{eq:Deltabound} and assumed that $r+\delta \gg 1$. Consider now the case that $r = N ^{\alpha}$ with $\alpha > 1/4$. We can then clearly choose $\eps, \delta$ depending on $r$ so as to have 
\beq r \gg \delta \gg \eps ^{1/2} \gg \left(N \log N \right) ^{1/2} r ^{-1},\eeq
which implies 
\beq N^{-1} r ^2 \gg \eps ^{-1} \log N \mbox{ and } 1 \gg \eps \delta ^{-2}\eeq
and thus~\eqref{avdensitybound bis} follows.
\end{proof}

\subsection{Localization procedure}\label{sec:localization}

We now want to extend Theorem~\ref{thm:densitybound} by dropping the angular momentum constraint. We use a rather standard localization procedure, see e.g.~\cite{Lewin-11,Rougerie-LMU} for its discussion in more general contexts. In order for the localization method to combine efficiently with what we have proved so far, we need the following elementary lemma:

\begin{lemma}[\textbf{The PL class}]\label{lem:PL class}\mbox{}\\
The class of functions $G$ satisfying~\eqref{eq:PL class} is closed under addition. Thus any function $G$ that can be written in the manner $G = \sum_{j} \alpha_j |F_j|$ with $\alpha_j$ positive numbers and $F_j$ analytic functions satisfies~\eqref{eq:PL class}.
\end{lemma}

\begin{proof}
This is a well-known fact, see e.g.~\cite[Paragraphs 2.12 to 2.14]{Rado-37} where~\eqref{eq:PL class} is called the PL class. One can verify it by a simple computation: if $G,H$ satisfy~\eqref{eq:PL class} we have 
\beq \Delta \log G = \frac{G\Delta G - |\nabla G| ^2}{G^2} \geq 0\eeq
and the same for $H$, thus
\begin{align}
 (G+H) ^2 \Delta \log (G+H) &= (G+H) (\Delta G+ \Delta H) - |\nabla G + \nabla H|^2\nonumber\\
 &= G \Delta G - |\nabla G ^2| + H \Delta H - |\nabla H| ^2 + G \Delta H + H \Delta G - 2 \nabla G \cdot \nabla H\nonumber\\
 &\geq \frac{G}{H} | \nabla H | ^2 + \frac{H}{G} | \nabla G | ^2 - 2 \nabla G \cdot \nabla H\nonumber
 \\
 &\geq \left| \sqrt{\frac{G}{H}} \nabla H - \sqrt{\frac{H}{G}} \nabla G \right| ^2\geq 0.
\end{align}
 The rest of the statement follows because any $|F|$ with $F$ analytic satisfies~\eqref{eq:PL class} as remarked after Equation~\eqref{eq:subharm 1}. See also~\cite[Page 84]{Klimek-91}.
\end{proof}

We now proceed to extending Theorem~\ref{thm:densitybound} to states not necessarily satisfying~\eqref{eq:ang mom bound}. This will conclude the

\begin{proof}[Proof of Theorem~\ref{thm:main}]
For $\Psi_F \in \cL_\ell ^N$ define
\beq \mubf_F (z_1,\ldots,z_N)= | \Psi_F (z_1,\ldots,z_N)| ^2.\eeq 
 
\medskip 

\noindent\textbf{Localization procedure.} We define a smooth partition of unity on $\R ^2$
\begin{equation}\label{eq:partition}
\chi ^2 + \eta ^2 \equiv 1,  \quad \chi ^2 (z) = e ^{- \frac{|z-a| ^2}{L ^2}}
\end{equation}
with $L \to \infty$ a length scale to be optimized over. Consider the classical $M$-particle state\footnote{Strictly speaking it is not normalized, thus not a state.} $\mubf_M$ built from $\mubf$ by localizing exactly $M$ particles and integrating out the other:
\begin{equation}\label{eq:mubf k}
\mubf_M (z_1,\ldots,z_M) := \prod_{j=1} ^M \chi^2 (z_j) \int_{\R ^{2(N-M)}} \prod_{j=M+1} ^N \eta ^2 (z_j) \mubf_F (Z_N) dz_{M+1} \ldots dz_N.  
\end{equation}
Observe that, using 
$$ \prod_{j=1} ^N \left( \chi ^2 (z_j) + \eta ^2 (z_j)\right)= 1$$
and symmetry under particle exchange, 
\beq\mubf_F (Z_N) = \sum_{M= 0} ^N { N \choose M} \prod_{j=1} ^M \chi^2 (z_j) \prod_{j=M+1} ^N \eta ^2 (z_j) \mubf_F (Z_N) \eeq 
and thus the $L^2$-normalization of $\Psi_F$ implies
\begin{equation}\label{eq:norm loc}
\sum_{M= 0} ^N  { N \choose M} \int_{\R ^{2M}} \mubf_M = 1.
\end{equation}
Next, recalling~\eqref{eq:original density} and~\eqref{eq:marginals} we denote by $ \mubf_F ^{(1)} = N ^{-1} \rho_F$ the one-particle probability density of $\Psi_F$. By symmetry under particle exchange again, we have 
\begin{align}\label{eq:dens loc}
 \chi ^2 (z_1) \mubf_F ^{(1)} (z_1) &= \int_{\R ^{2 (N-1)}} \chi^2 (z_1) \mubf_F(Z_N) dz_2\ldots dz_N\nonumber\\
 &= \int_{\R ^{2 (N-1)}} \chi^2 (z_1) \prod_{M=2} ^{N} \left( \chi^2 (z_M) + \eta ^2 (z_M) \right) \mubf_F(Z_N) dz_2\ldots dz_N \nonumber\\
&= \sum_{M=1} ^N {N-1 \choose M-1} \int_{\R ^{2 (N-1)}} \chi^2 (z_1) \ldots \chi ^2 (z_M) \eta ^2 (z_{M +1}  )\ldots \eta ^2 (z_N)  \mubf_F(Z_N) dz_2\ldots dz_N \nonumber\\
&= \sum_{M=1} ^N {N-1 \choose M-1} \mubf_{M} ^{(1)} (z_1).
\end{align}
Finally, since clearly $|z-a|^2 \chi^2 (z) \leq L ^{2+\delta} \chi^2 (z)$ for arbitrary $\delta >0$ we have 
\begin{align}\label{eq:bound mom loc}
\sum_{M= 1} ^N  { N-1 \choose M - 1} &\int_{\R ^{2}} |z_1-a| ^2 \mubf_M ^{(1)} (z_1) dz_1  \nonumber\\
&= \sum_{M= 1} ^N  { N -1 \choose M -1 } \int_{\R ^{2M}} |z_1-a| ^2  \prod_{j=1} ^M \chi^2 (z_j) \prod_{j=M+1} ^N \eta ^2 (z_j) \mubf_F(Z_N) dz_1\ldots dz_N \nonumber\\
&\leq L ^{2+\delta} \int_{\R ^{2M}} \sum_{M= 1} ^N  { N -1 \choose M -1 } \prod_{j=2} ^M \chi^2 (z_j) \prod_{j=M+1} ^N \eta ^2 (z_j) \mubf_F(Z_N) dz_1\ldots dz_N\nonumber\\
&\leq L ^{2+\delta} \int_{\R ^{2M}} \prod_{j=2} ^{N} \left( \chi^2 (z_j) + \eta ^2 (z_j) \right) \mubf_F(Z_N) dz_1 \ldots dz_N = L ^{2+\delta}. 
\end{align}

\medskip 

\noindent\textbf{Application.} Pick a disk $D$ of center $a\in\R^2$ and radius $r = N^{\alpha}$ with $\alpha > 1/4$ and choose
\beq \chi^2 (z) = \exp \left( -\frac{|z-a| ^2}{L^2}\right)\eeq 
with $L = N ^{\beta}, \beta > \alpha$. Then clearly 
\begin{equation}\label{eq:first step}
\int_{D} \mubf_F ^{(1)} \leq \left( 1+o(1) \right) \int_{D} \chi ^2 \mubf_F ^{(1)}
\end{equation}
and we expect the right-hand side to be at most of order $N ^{-1 + 2 \alpha}$. To estimate it, we start from~\eqref{eq:dens loc} and first brutally get rid of terms with small $M$ for they have too few particles to contribute to the density at this scale: 
\begin{align}\label{eq:low terms}
 \int_{D} \chi ^2 \mubf_F ^{(1)} &= \sum_{M=1} ^{N } {N \choose M} \frac{M}{N} \int_D \mubf_{M} ^{(1)} \nonumber\\
&\leq N ^{\gamma-1} \sum_{M=1} ^{N^{\gamma}} {N \choose M} \int_D \mubf_{M} ^{(1)} + \sum_{M=N^{\gamma}} ^N {N \choose M} \frac{M}{N} \int_D \mubf_{M} ^{(1)} \nonumber\\
&\leq N ^{\gamma-1} \sum_{M=1} ^{N^{\gamma}} {N \choose M} \int_{\R ^{2M}} \mubf_{M}  + \sum_{M=N^{\gamma}} ^N {N \choose M} \frac{M}{N} \int_D \mubf_{M} ^{(1)}\nonumber\\
&\leq N ^{\gamma-1} + \sum_{M=N^{\gamma}} ^N {N \choose M} \frac{M}{N} \int_D \mubf_{M} ^{(1)}
\end{align} 
where we used that 
\begin{equation}\label{eq:combinatorics}
 {N \choose M} = \frac{N}{M} {N-1 \choose M-1} 
\end{equation}
and inserted~\eqref{eq:norm loc}. We pick $\gamma < 2 \alpha$ to render the first term $N^{\gamma-1}$ in the above much smaller than the expected $N^{-1 + 2 \alpha}$, so that we may focus on the second one. 

We thus need to bound the $1$-particle probability density $\mubf_M ^{(1)}$, first marginal of the $M$-particles state $\mubf_M$. Recalling~\eqref{eq:mubf k} we can write  
\begin{equation}\label{eq:mubf k bis}
\mubf_M (z_1,\ldots,z_k) = \left|\PsiLau (z_1,\ldots,z_M) \right|^2 \exp\left( - \sum_{j=1} ^M \frac{|z_j-a|^2}{L^2} \right) G_M (z_1,\ldots,z_M)   
\end{equation}
where ($c_M$ is a positive constant ensuring normalization of~\eqref{eq:mubf k bis})
\begin{multline}\label{eq:GM}
 G_M (z_1,\ldots,z_M) =  c_M \int_{\R ^{2(N-M)}} \prod_{j=M+1} ^N \eta ^2 (z_j) e^{-|z_j| ^2} \prod_{M+1 \leq i<j \leq N} |z_i-z_j| ^{2\ell} \\ |F(z_1,\ldots,z_N)| ^2  \prod_{i= 1} ^M  \prod_{j=M+1} ^N |z_i-z_j| ^{2\ell} dz_{M+1} \ldots dz_N.  
\end{multline}
At fixed $z_{M+1},\ldots,z_N$, the integrand in the above is the squared modulus of an analytic function of the variables $z_1,\ldots,z_M$ (they appear only in the second line) and is thus subharmonic. Since the PL class is closed under addition, c.f. Lemma~\ref{lem:PL class}, we deduce that $G_M$ satisfies Assumption~\eqref{eq:PL class}. We can thus apply~\eqref{avdensitybound bis} to the normalized version of $\mubf_M$ and deduce 
\begin{multline}\label{avdensitybound bis M}
\frac{\int_D\mubf_M ^ {(1)}}{\int_{\R ^{2M}} \mubf_M } \leq \frac 1{M\ell\pi}|D| \left(1+o_D (1) + O \left(L^{-2} \right)  \right) 
\\ + C M ^{-1 - 2 \alpha'} \left(\int_{\R ^{2M}} \mubf_M \right)^{-1}\intR |z-a| ^2 \mubf_M ^{(1)} (z)dz  
\end{multline}
for any $\alpha'<\alpha$. Next, we sum the above inequalities from $M= N ^{\gamma}$ to $N$ and insert the result in~\eqref{eq:low terms} to obtain
\begin{align}\label{eq:high terms}
 \int_{D} \chi ^2 \mubf_F ^{(1)} &\leq \sum_{M=N^{\gamma}} ^N {N \choose M} \frac{M}{N} \left( \int_{\R^{2M}} \mubf_{M}\right) \frac{|D|}{\pi \ell M}\left( 1 + o _N(1) \right) \nonumber\\
 &+ C N^{-\gamma - 2 \gamma \alpha'} \sum_{M=N^{\gamma}} ^N {N \choose M} \frac{M}{N} \intR |z-a| ^2 \mubf_M ^{(1)} (z)dz + N ^{\gamma - 1}  
\end{align}
where we have used that we sum only terms where $M\gg 1$ to ensure that the $o_D (1)$ terms are indeed small. We next insert~\eqref{eq:norm loc} in the first sum and use~\eqref{eq:combinatorics} in combination with~\eqref{eq:bound mom loc} to deal with the second one, obtaining
\begin{equation}\label{eq:high terms 2}
 \int_{D} \chi ^2 \mubf_F ^{(1)} \leq \frac{|D|}{\pi \ell N}\left( 1 + o _N(1) \right) + C N^{-\gamma - 2 \gamma \alpha ' + (2+\delta) \beta  } + N^{\gamma - 1}.
 \end{equation}
To optimize over $\gamma$ we take 
\beq \gamma = \frac{1}{2+2\alpha'} + \frac{(1+\delta)\beta}{2+2\alpha'}\eeq
and since we are at liberty to choose $\alpha',\beta$ arbitrarily close to $\alpha$ and $\delta$ arbitrarily small, we can render $\gamma$ arbitrarily close to $1/2$, so that the error terms in~\eqref{eq:high terms 2} become $O (N ^{- 1/2 + \eps})$ for arbitrarily small $\eps >0$. For $\alpha > 1/4$ we deduce
\beq\int_D \mubf_F^{(1)} \leq (1+o(1)) \int_{D} \chi ^2 \mubf_F ^{(1)} \leq \frac{|D|}{\pi \ell N}\left( 1 + o _N(1) \right)\eeq
as desired because the main term is of order $N ^{-1 + 2 \alpha}$. To deduce the general result~\eqref{eq:gen dens bound} we can simply replace indicative functions of disks by indicative functions of more general open sets everywhere in Section~\ref{sec:dens Gibbs}, relying on~\eqref{eq:pert dens bound} for a general open set.
\end{proof}

\section{Energy in confining potentials}\label{sec:ener low bound}

We now provide the 

\begin{proof}[Proof of Corollary~\ref{cor:pot ener}]
This is a variation on the arguments of Section~\ref{sec:cond bound}, but it is more convenient to work with scaled variables as in~\cite[Section~3]{RouYng-15}. We thus write, with the notation of Assumption~\ref{asum:pot} 
\begin{equation}
\intR V \rhoF  = N \intR U  \muF ^{(1)}  
\end{equation}
with the scaled one-particle probability density
\begin{equation}\label{eq:rescale 1p dens}
\muF ^{(1)}  (x) = \rhoF \left( \sqrt{N} x \right). 
\end{equation}
We will deduce the result from the lower bound 
\begin{equation}\label{eq:low bound rescaled}
\intR U  \muF ^{(1)} \geq (1+o(1)) \inf\left\{ \intR U \mu \: | \: 0\leq \mu \leq \frac{1}{\pi \ell}, \: \intR \mu = 1 \right\} + o (1)
\end{equation}
by a simple change of variables. 

Let $L>R$ be a large but fixed number, with $R$ as in Assumption~\ref{asum:pot}. Define a $U_L \leq U$ such that
\begin{equation}
U_L  = \begin{cases}
        U \mbox{ in } D(0,L) \\
        \underline{U} \mbox{ in } D(0,2L)^c
       \end{cases}
\end{equation}
and $U_L$ satisfies the same bounds~\eqref{eq:pot infi 1}-~\eqref{eq:pot infi 2} as $\underline{U}$ outside of $D(0,L)$. We then consider (compare with~\eqref{eq:scale Hamil}) 
\begin{multline}\label{eq:scale Hamil pert}
\tilde{H}_N ^\eps (z_1,\ldots,z_N) = \sum_{j=1} ^N \left(|z_j| ^2 + \eps U_L (z_j) \right)- 2 \frac{\ell}{N} \sum_{1\leq i<j\leq N} \log |z_i-z_j| 
\\ - \frac{2}{N} \log |F ( \sqrt{N} z_1,\ldots, \sqrt{N}z_N)| 
\end{multline}
and estimate (recall the notation~\eqref{eq:scale Gibbs})
\begin{equation}
\cEt_N ^{\eps} [\muF] := \intRN \tilde{H}_N ^\eps \muF + \frac{1}{N} \intRN \muF \log \muF.
\end{equation}
Since $\muF$ is the Gibbs state at temperature $1/N$ associated with $\tilde{H}_N ^\eps$ and $\muF ^{(1)}$ its first marginal we have 
\begin{align}\label{eq:free ener up again}
\cEt_N ^{\eps} [\muF] &= \eps N \intR U_L \muF ^{(1)} + \intRN \tilde{H}_N ^0 \muF + \frac{1}{N} \intRN \muF \log \muF \nonumber\\
&\leq \eps N \intR U_L \muF ^{(1)} + \min_{\R^{2N}} \tilde{H}_N ^0 + C (\log N + 1)
\end{align}
using a trial state similar to~\eqref{eq:trial state} to bound the minimal free energy at temperature $1/N$ associated with $\tilde{H}_N ^0$, exactly as in~\cite[Section~3.1]{RouYng-15}. 

On the other hand, by Jensen's inequality  
\begin{align}
\intRN \muF \log \muF &\geq \intRN \muF \log \left( \frac{e^{-U_L}}{\intR e^{-U_L} }\right) ^{\otimes N}\nonumber\\
&\geq - N \intR U_L \muF ^{(1)} - C N 
\end{align}
where we used~\eqref{eq:pot infi 2} to bound $\intR e^{-U_L}$. Thus 
\begin{align}
\cEt_N ^{\eps} [\muF] &\geq \min_{\R^{2N}} \left\{ \tilde{H}_N ^\eps (z_1,\ldots,z_N) - \frac{1}{N} \sum_{j=1} ^N U_L (z_j)\right\} - C \nonumber \\
&\geq \min_{\R^{2N}}  \tilde{H}_N ^0 + \left(N\eps - 1 \right) \intR U_L \tilde{\mu}^0_\eps -C
\end{align}
where 
\beq 
\tilde{\mu}^0_\eps(z)=\frac 1N\sum_{i=1}^N\delta(z-z_i^0)
\eeq 
is the empirical measure of a configuration $(z_1 ^0,\ldots, z_N ^0)$ minimizing the classical Hamiltonian $\tilde{H}_N ^\eps (z_1,\ldots,z_N) - \frac{1}{N} \sum_{j=1} ^N U_L (z_j)$. We can apply a (re-scaled version of) Proposition~\ref{pro:dens plasma} to $\tilde{\mu}^0_\eps$. Thus, for any domain $\Omega_r$ obtain from a fixed Lipschitz $\Omega$ by scaling lengths by a factor $r = N ^{\alpha'}$, $\alpha' > -1/2$ 
\begin{align} \label{eq:pert dens bound again}
\int_{\Omega_r} \tilde{\mu}^0_\eps &\leq \frac 1{\ell \pi}|\Omega_r|\left( 1 + \frac {\eps}4 \Vert \Delta U_L\Vert_{\infty} \right)(1+o(1)) \nonumber\\
&\leq \frac 1{\ell \pi}|\Omega_r|\left( 1 + C \eps N ^{1 - 2\alpha} \right)(1+o(1))
\end{align}
where we used~\eqref{eq:pot vari} to estimate $\Vert \Delta U_L\Vert_{\infty}$. By a Riemann sum approximation on a grid of side length $N ^{\alpha'}$ we then have 
\begin{equation}
\intR U_L \tilde{\mu}^0_\eps \geq \intR U_L \hat{\mu}^0_\eps - C \norm{\nabla U_L}_{L^{\infty}} N ^{\alpha'}
\end{equation}
where $\hat{\mu}^0_\eps$ is a piecewise constant function satisfying 
\begin{equation}
0 \leq \hat{\mu}^0_\eps \leq \frac{1}{\ell \pi} \left( 1 + C \eps N ^{1-2\alpha} \right)(1+o(1))
\end{equation}
and 
\begin{equation}
\intR  \hat{\mu}^0_\eps = 1.
\end{equation}
Since, using~\eqref{eq:pot vari}, $\norm{\nabla U_L}_{L^{\infty}} \leq C N ^{1/2-\alpha}$ we deduce
\begin{equation}
\cEt_N ^{\eps} [\muF] \geq \min_{\R^{2N}}  \tilde{H}_N ^0 + \left(N\eps - 1 \right) \intR U_L \hat{\mu}^0_\eps - C N^{1/2 + \alpha' - \alpha}
\end{equation}
for any fixed $\alpha > - 1/2$. Comparing with~\eqref{eq:free ener up again} we obtain  
\begin{align}
\intR U \muF ^{(1)} &\geq \left( 1 - \frac{1}{\eps N} \right) \intR U_L \hat{\mu}^0_\eps - \frac{C \log N}{\eps N} - \frac{C N ^{1/2 + \alpha' - \alpha}}{\eps N}  \nonumber\\
&\geq \left( 1 - \frac{1}{\eps N} \right) \inf\left\{ \intR U_L \mu \: | \: 0\leq \mu \leq \frac{1}{\pi \ell} \left( 1 + C \eps N ^{1 - 2 \alpha} \right)(1+o(1)), \: \intR \mu = 1 \right\} \nonumber\\
&- \frac{C \log N}{\eps N} - \frac{C N ^{1/2 + \alpha' - \alpha}}{\eps N}  
\end{align}
where we used that $U \geq U_L$. Since $\alpha > 0$ and $\alpha'$ can be chosen  arbitrarily close to $-1/2$ in the above, we can choose $\eps = N ^{-\delta}$ with some $\delta$ less than, but sufficiently close to, $1$ to get 
\begin{equation}
 \intR U \muF ^{(1)} \geq \left( 1 - o(1) \right) \inf\left\{ \intR U_L \mu \: | \: 0\leq \mu \leq \frac{1}{\pi \ell}, \: \intR \mu = 1 \right\} + o(1)
\end{equation}
The continuity of the bathtub energy as a function of the upper density constraint we have just used is a consequence of the explicit solution of the minimization problem~\cite[Theorem~1.14]{LieLos-01}. To conclude the proof of~\eqref{eq:low bound rescaled} we note that, under our assumptions on $U$, it suffices to take $L$ large enough but fixed to have 
\begin{equation}
\inf\left\{ \intR U_L \mu \: | \: 0\leq \mu \leq \frac{1}{\pi \ell}, \: \intR \mu = 1 \right\} = \inf\left\{ \intR U \mu \: | \: 0\leq \mu \leq \frac{1}{\pi \ell}, \: \intR \mu = 1 \right\}. 
\end{equation}
This follows from the fact that minimizers of the bathtub energy are equal to either $0$ or $(\pi\ell )^{-1}$ almost everywhere, and that for potentials increasing at infinity, their support is compact.
\end{proof}

\section{Conclusions}

By employing a new theorem about multi-particle screening we have derived sharp local upper bounds on the particle density in  ground states of a classical 2D Coulomb gas perturbed by arbitrary many-body potentials that are superharmonic in every variable. We have also derived corresponding bounds for Gibbs states which, via Laughlin's plasma analogy, have direct applications to quantum many-particle states in strong magnetic fields exhibiting the Fractional Quantum Hall Effect (FQHE) and which are fully correlated in order to minimize a repulsive interaction between the particles. Our results hold asymptotically for large particle numbers, on length scales that are small compared to the full system's extension. We expect that there is room for improvement as regards the smallest length scale we can afford. This remains an \emph{open problem}.

The bounds for the classical Coulomb ground states rest on two pillars: First, an \lq\lq exclusion rule\rq\rq\ for ground state configurations which prevents points in such a configuration from entering a screening region defined by the other points in the configuration. Second, a universal density bound that we prove for any configuration of points satisfying the exclusion rule.

For the Gibbs state the modification of the ground state density caused by the entropy term in the free energy has to be estimated. This is done by perturbing the classical Coulomb Hamiltonian locally and bounding the perturbed free energy, obtaining a density bound via the Feynman-Hellmann principle. This procedure is first carried out under the assumption of an a-priori bound on the angular momentum and then extended to the general case by a localization argument.

Besides the density bounds we also derive lower bounds for the energy of fully correlated FQHE states confined in external potentials. In combination with recently obtained upper bounds~\cite{RouYng-17}, this establishes that approximate quantum mechanical ground states in the class of strongly correlated FQHE wave-functions can be obtained by generating uncorrelated quasi-holes on top of the Laughlin state.

\bibliographystyle{siam}

\end{document}